%% file: main.tex
\documentclass{article} 
\usepackage{arxiv,times}

\input{math_commands.tex}

\input{preamble.tex}
\usepackage{hyperref}
\usepackage[utf8]{inputenc} 
\usepackage[T1]{fontenc}    
\usepackage{hyperref}       
\usepackage{url}            
\usepackage{booktabs}       
\usepackage{amsfonts}       
\usepackage{nicefrac}       
\usepackage{microtype}      
\usepackage{xcolor}         
\usepackage{amsmath,amsthm,thmtools}
\usepackage{amsfonts}
\usepackage{amssymb}

\newcommand{\uni}[2]{\mathrm{Uni}({#1{:}#2})}
\newcommand{\rdn}[2]{\mathrm{Red}({#1{:}#2})}
\newcommand{\syn}[2]{\mathrm{Syn}({#1{:}#2})}
\newcommand{\mut}[2]{\mathrm{I}({#1;#2})}
\newcommand{\mutd}[3]{\mathrm{I}_{#1}({#2;#3})}
\usepackage{wrapfig}
\newcommand{\iid}[0]{i.i.d.}
\newcommand{\given}[0]{|}
\usepackage{etoc}
\newtheorem{lem}{Lemma}
\newtheorem{thm}{Theorem}
\newtheorem{defn}{Definition}
\newtheorem{prop}{Proposition}

\usepackage{enumitem}
\usepackage{subcaption}
\usepackage{booktabs}
\usepackage{multirow}

\usepackage{natbib}
\usepackage[most]{tcolorbox}
\usepackage[table]{xcolor}
\usepackage{array}
\usepackage{graphicx}
\usepackage{svg}
\newtcolorbox{myquote}{
  colback=blue!3,        
  colframe=blue!60!black,
  leftrule=3pt,          
  arc=0pt,               
  boxrule=0pt,           
  enhanced,
  breakable,
  left=6pt,
  right=6pt,
  top=6pt,
  bottom=6pt
}

\usepackage[most]{tcolorbox}
\usepackage{xcolor}

\definecolor{findingblue}{RGB}{70, 130, 180}
\definecolor{findingbackground}{RGB}{246, 250, 254}

\newtcolorbox{finding}{
    enhanced,
    breakable,
    colback=findingbackground,
    colframe=findingblue,
    boxrule=1pt,
    arc=5mm,
    outer arc=5mm,
    left=5mm,
    right=5mm,
    top=3.5mm,
    bottom=3.5mm,
    boxsep=0mm,
    before skip=10pt,
    after skip=10pt
}
\title{Interpreting Reasoning of Large Language Models via Partial Information Decomposition}

\author{Barproda Halder \thanks{Corresponding author.} \\
University of Maryland, College Park\\
\texttt{bhalder@umd.edu} \\
\And
Qiuyi Zhang \\
Elorian AI\\
\texttt{richard@elorian.ai} \\
\AND
Sanghamitra Dutta \\
University of Maryland, College Park\\
\texttt{sanghamd@umd.edu}
}

\iclrfinalcopy 
\begin{document}

\maketitle

\begin{abstract}


Large reasoning models (LRMs) have achieved substantial improvements in solving complex mathematical problems, but often produce lengthy, repetitive, or erroneous reasoning trajectories. In this work, we introduce a new interpretability framework \strategy to evaluate the quality of the reasoning process. \strategy leverages an emerging body of work from information theory called Partial Information Decomposition to disentangle the information about the final answer between two consecutive reasoning steps into non-negative components: unique information (in preceding steps or current step), redundant information, and synergistic information. Building on this decomposition, we propose the \emph{Step-wise Repetitive Reasoning Index (Step-RRI)}, a theoretically grounded measure that assesses whether the answer-relevant information in the current step $S_i$ is predominantly redundant with the past steps $S_{<i}$, relative to its unique and synergistic contributions. To evaluate the effectiveness of Step-RRI in detecting repetitiveness, we apply \strategy to the redundancy class of the PRMBench dataset where Step-RRI improves step-level redundancy identification accuracy by over $10$ points compared to embedding-similarity and information-gain baselines. Next, we define \emph{Trajectory-RRI}, an aggregate measure of repetitiveness for an individual reasoning trajectory. To demonstrate its practical relevance, we show that average Trajectory-RRI strongly correlates with actual reasoning length across QwQ-32B, DeepSeek-R1-Distill-Qwen-32B, and GPT-4.1, motivating its use as a signal for improving reasoning efficiency. Finally, we introduce \emph{Trajectory-RRI-guided data selection for fine-tuning}, demonstrating that selecting training data based on Trajectory-RRI can improve a fine-tuned model's reasoning efficiency while largely preserving its task performance. Our results establish RRI measures as both a step-level diagnostic of repetitive reasoning and a practical signal for improving reasoning efficiency. 

  
\end{abstract}

\section{Introduction}

The emergence of large reasoning models (LRMs), such as OpenAI's o1~\citep{OpenAI}, DeepSeek-R1~\citep{guo2025deepseek}, and QwQ-32B~\citep{team2025qwq}, has marked a significant leap in complex problem solving. These models typically engage in extensive reasoning or thinking before producing a final answer, often generating long chain-of-thought (CoT) reasoning trajectories~\citep{wei2022chain}. However, such reasoning can be excessively verbose and repetitive, frequently revisiting information that has already been established in earlier steps. This redundant reasoning increases computational cost without necessarily improving the correctness of the final answer. Moreover, as reasoning trajectories become longer and more convoluted, they become increasingly difficult for humans to follow and evaluate, making it challenging to determine whether individual steps are correct, informative, or meaningfully contribute to the final answer. \emph{Thus, there is a pressing need for automated methods to evaluate the quality of individual reasoning steps without human supervision.}

To address this challenge, we propose \strategy (\textbf{SL}iding Window \textbf{I}nformation \textbf{D}ecomposition \textbf{E}xplainer for \textbf{R}easoning), a unified interpretability framework that leverages Partial Information Decomposition (PID) to analyze the quality of individual reasoning steps. Building on PID, we introduce a novel metric \emph{Step-wise Repetitive Reasoning Index (Step-RRI)}, which uses the relative contributions of redundant, unique, and synergistic answer-relevant information to characterize whether an individual reasoning step is repetitive with respect to the preceding reasoning. Fig.~\ref{fig:toy-rri} illustrates an example of reasoning steps exhibiting different degrees of repetitiveness, and hence varying Step-RRI values.\begin{wrapfigure}[18]{r}{0.5\columnwidth}  
  \centering
  \includegraphics[width=\linewidth]{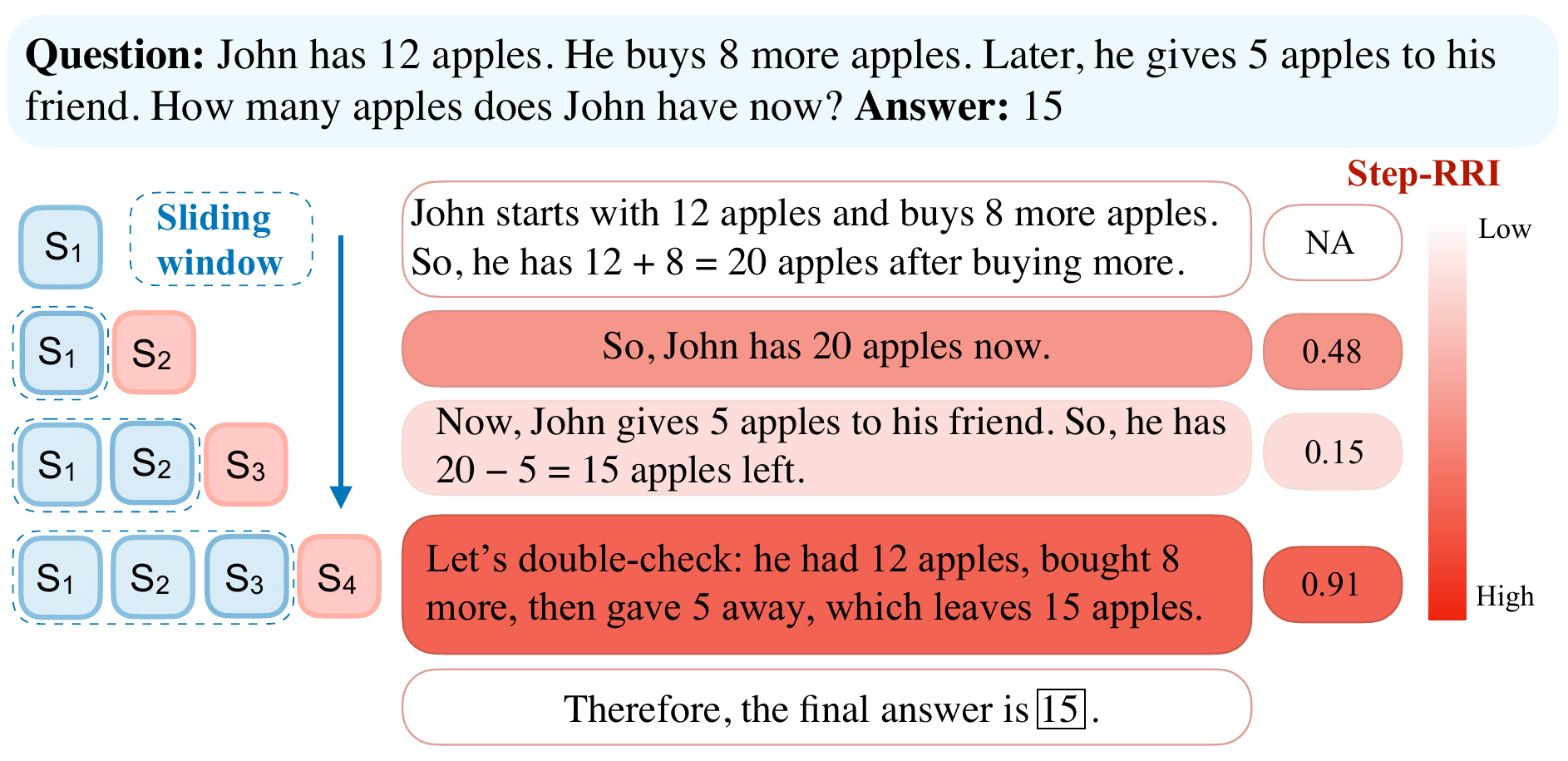}
  \caption{\emph{Example illustrating the role of \strategy in identifying repetitiveness.} The sliding window compares each reasoning step with the preceding steps. A step that adds answer-relevant information, such as computing the remaining apples, has a low expected Step-RRI; a step that repeats the same calculation has a high expected Step-RRI.}
  \label{fig:toy-rri}
\end{wrapfigure} To quantify reasoning repetitiveness at the trajectory level, we further define \emph{Trajectory-RRI} by aggregating repetitive behavior across the individual steps of a reasoning trajectory. Together, these measures characterize reasoning repetitiveness at the step and trajectory levels without requiring reference reasoning traces.

Recent works~\citep{ton2024understanding,yong2025think} have analyzed reasoning using classical information-theoretic quantities such as mutual information and conditional mutual information. However, classical information-theoretic measures only quantify whether the current step is informative or not, without distinguishing whether the steps contribute jointly or individually to the final answer. Such aggregate information gains fail to capture whether a step contributes uniquely, has redundancy, or has complementary information with the preceding reasoning. In contrast, \strategy characterizes repetitiveness directly by decomposing the information between the current step and its past reasoning steps about the final answer, enabling reference-free, step-level analysis of repetitive reasoning. Our main contributions can be summarized as follows:

\noindent \textbf{Disentangling Step-wise Reasoning through PID:} We propose a principled interpretability framework called \strategy to analyze when and how individual reasoning steps contribute \begin{wrapfigure}[14]{r}{0.4\columnwidth}  
  \vspace{-5pt}
  \centering
  \includegraphics[width=0.9\linewidth]{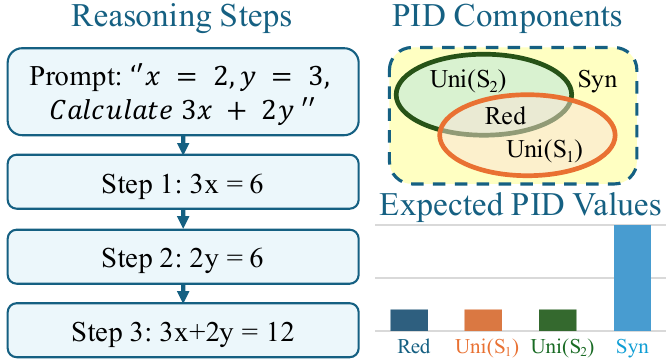}
  \caption{\emph{\strategy in a simple arithmetic problem:} Demonstrates high synergistic information between Step 1 (3x) and Step 2 (2y) because they jointly provide complete information about (3x+2y).}
  \label{fig:summary}
\end{wrapfigure}  to the final answer. \strategy is rooted in Partial Information Decomposition (PID), which decomposes the information about the final answer between two reasoning steps into unique, redundant, and synergistic components (see Fig.~\ref{fig:summary} for an illustrative example).
 Given the reasoning trajectories, we move across the steps in a sliding-window fashion and project them onto a meaningful embedding space. \strategy discretizes the embeddings of the reasoning steps and the final answer and estimates the PID measures. To examine the effect of noisy reasoning steps on these PID components, we consider two consecutive reasoning steps, $S_{i-1}$ and $S_i$, and inject independent synthetic noise into an individual step. We show, both theoretically (Theorem~\ref{thm:uniqueness_noise}) and empirically on synthetic arithmetic data, that the unique information contributed by the perturbed step consistently decreases.

\noindent \textbf{Quantifying Repetitive Reasoning through RRI:} We introduce the \emph{ Step-wise Repetitive Reasoning Index (Step-RRI)}, an information-theoretic measure for characterizing repetitive reasoning at the step level. \emph{Step-RRI} characterizes whether the answer-relevant information contributed by the current step is predominantly redundant with the preceding reasoning relative to its unique and synergistic contributions (Proposition~\ref{prop:red}). Furthermore, to quantify the overall repetitiveness of a reasoning trajectory, we define \emph{Trajectory-RRI} (Definition~\ref{def:point-wiseRRI}) by aggregating repetitive behavior across individual steps. We theoretically demonstrate (see Theorem~\ref{thm:repititive2}) how Step-RRI effectively captures repetitiveness. Our proof leverages monotonicity properties of the PID measure. We empirically validate the effectiveness of Step-RRI in detecting actual repetitiveness using the annotated redundancy class of PRMBench~\citep{song2025prmbench}. Step-RRI achieves a 10-point improvement in accuracy for step-level redundancy detection, outperforming embedding-similarity and information-gain baselines.

\noindent \textbf{Linking Trajectory-RRI to Reasoning Efficiency:} To assess the practical utility of Trajectory-RRI for capturing reasoning efficiency, we examine its relationship with actual reasoning length 
across multiple LLMs (see Table~\ref{tab:rri_correlation}). We evaluate reasoning trajectories on GSM8K~\citep{cobbe2021gsm8k} and AIME2025~\citep{aime25} from three models QwQ-32B~\citep{team2025qwq}, DeepSeek-R1-Distill-Qwen-32B~\citep{guo2025deepseek}, and GPT-4.1~\citep{openai2025gpt41}. For these models, Trajectory-RRI is strongly correlated with reasoning length (Spearman’s correlation coefficient, $\rho=0.91$, $0.95$, and $0.73$, respectively, on GSM8K). We find a similarly strong association between Trajectory-RRI and reasoning length on AIME2025~\citep{aime25}. This consistent relationship indicates that reasoning trajectories that are more repetitive, as determined by our measure, also tend to be longer and thus less efficient, highlighting a practical utility of Trajectory-RRI. 

\noindent \textbf{Leveraging Trajectory-RRI for Efficient Data Selection:}
Beyond post-hoc analysis, we demonstrate that Trajectory-RRI can be directly incorporated into reasoning data selection for supervised fine-tuning (SFT). 
Using GSM8K~\citep{cobbe2021gsm8k}, we construct training datasets with different levels of average Trajectory-RRI by selecting reasoning trajectories generated by QwQ-32B~\citep{team2025qwq} and DeepSeek-R1-Distill-Qwen-32B~\citep{guo2025deepseek}. We use these datasets to fine-tune Qwen2.5-7B-Instruct~\citep{qwen2.5} with Low-Rank Adaptation (LoRA)~\citep{hu2021lora}. We find that the efficiency of the fine-tuned models, evaluated using an LLM judge, is strongly negatively correlated ($\rho=-0.97$) with the average Trajectory-RRI of the training data, while task performance remains largely preserved. These findings establish Trajectory-RRI as an actionable signal for selecting reasoning trajectories during fine-tuning, providing a practical approach to training models to reason more efficiently without sacrificing task performance.

\subsection{Related Works} \textit{Classical information theoretic measures} offer a lens to theoretically analyze and evaluate different aspects of LLMs ~\citep{farquhar2024detecting,kossen2024semantic,ali2025entropy,ton2024understanding,gan2025rethinking,yong2025think}. Recent works \citet{gan2025rethinking, ton2024understanding,gan2025rethinking,yong2025think} use information-theoretic measures to evaluate LLM reasoning. \citet{gan2025rethinking} explore slow-thinking in the multi-step reasoning of LLMs using information theory. \citet{ton2024understanding} use conditional mutual information to quantify stepwise information gain and identify failure modes in CoT reasoning. \citet{yong2025think} introduce two metrics, namely InfoGain and InfoBias, to quantify stepwise information gain and divergence from the original reasoning path, respectively, utilizing classical information-theoretic measures. \emph{Although these works demonstrate the utility of information-theoretic measures for analyzing LLM reasoning, the structural interactions between reasoning steps, particularly their joint and individual contributions to the final answer, remain less explored. Also, theory-grounded approaches for detecting step-level repetitiveness without reference reasoning trajectories remain underexplored.}

\textit{Partial Information Decomposition} (PID)~\citep{williams2010nonnegativePID, venkatesh2024gaussian, goswami2023computing, pakman2021estimating, lyu2024explicit} is an emerging area of research of information theory, with growing applications across diverse domains, including representation analysis, fairness, multimodal learning, and the analysis of large language models (LLMs), vision-language models (VLMs), and diffusion models~\citep{tax2017partial, dutta2020information, dutta2021fairness,hamman2024demystifying, hammanUnified, ehrlich2022partial, liang2023multimodal, wollstadt2023rigorous, mohamadi2023more, diffusionPID, dissanayake2024quantifying, haldertowards,halder2025learning,xiu2026comprehensive}. For instance, PID has been used to quantify fairness in \citet{dutta2020information, dutta2021fairness,hamman2024demystifying, hammanUnified} (also see \citet{entropy_survey} for a survey). \citet{haldertowards} propose an explainability framework leveraging PID to disentangle the nature of spurious associations in a dataset. Beyond dataset analysis, \citet{halder2025learning} and \citet{dissanayake2024quantifying} incorporate objectives based on redundant information into their optimization frameworks, for invariant graph learning and knowledge distillation, respectively. \citet{liang2023multimodal} employ PID to analyze multimodal interactions, while \citet{diffusionPID} leverage PID to interpret diffusion models. \citet{xiu2026comprehensive} measure the “information spectrum” of Large vision-language models (LVLMs) by decomposing decision-relevant information into redundant, unique, and synergistic components. \emph{Despite these advances, the application of PID to reasoning processes, particularly to analyze the informational structure of intermediate reasoning steps, remains largely unexplored. In this work, we leverage PID to the study of structured reasoning, introducing a principled framework to quantify repetitiveness at the step level and then at the trajectory level.}

\subsection{Background on PID}\begin{wrapfigure}[9]{r}{0.29\columnwidth}  
\vspace{-15pt}
  \centering
  \includegraphics[width=0.9\linewidth]{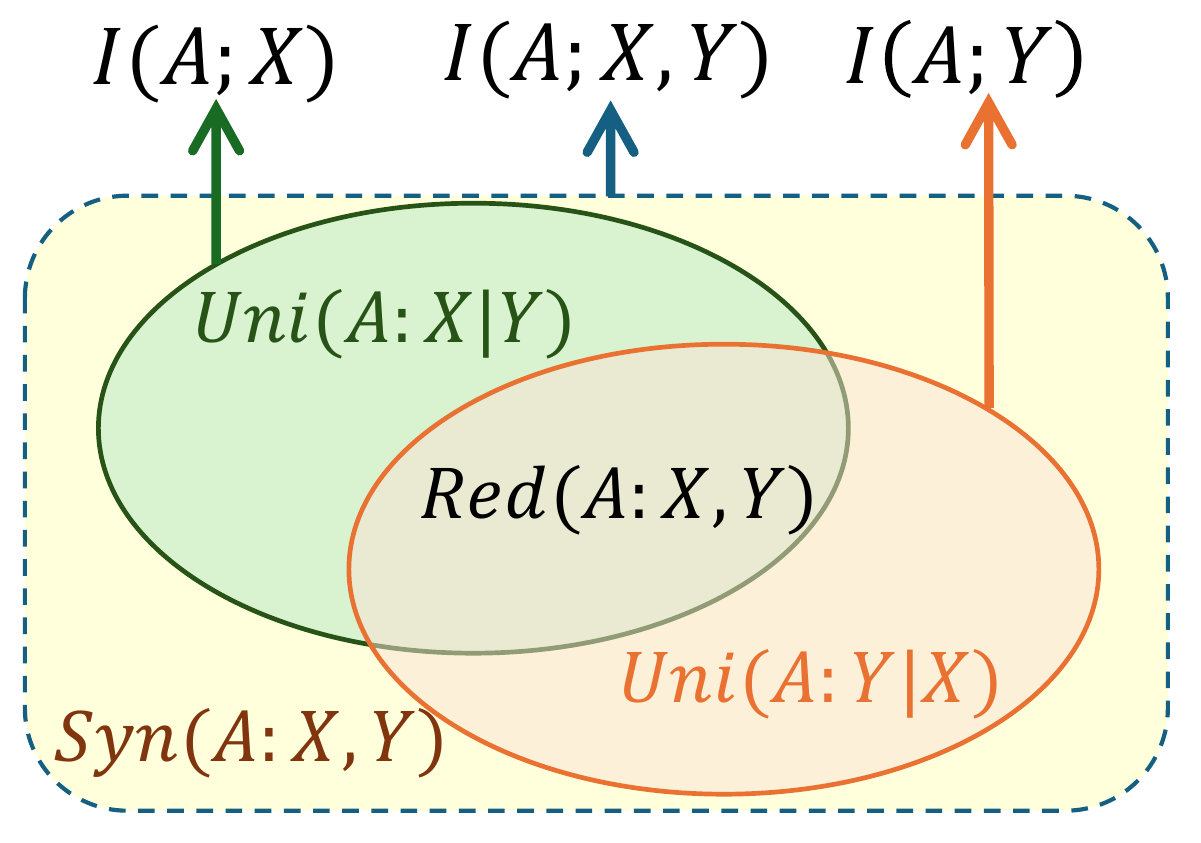}
  \caption{PID of $\mut{A}{X,Y}$}
  \label{fig:PID}
\end{wrapfigure} The classical measure of the total information about a target variable $A$ that is contained in two random variables $X$ and $Y$ is given by mutual information $\mut{A}{X,Y}$~\citep{cover2012elements}. However, mutual information $\mut{A}{X,Y}$ does not disentangle what is uniquely contributed by each or shared by both. An emerging body of work called Partial Information Decomposition (PID) goes beyond classical measures and decomposes the joint information about a target $A$ among multiple random variables $X$ and $Y$ into four non-negative measures (see Fig. \ref{fig:PID}). Here,  unique information $\uni{A}{X|Y}$ and $\uni{A}{Y|X}$ capture the information that is exclusively provided by $X$ or $Y$. Redundant information $\rdn{A}{X,Y}$ is the shared information about $A$ in $X$ and $Y$, and synergistic information $\syn{A}{X,Y}$ emerges only when $X$ and $Y$ are both present together (complementary). We now formally define these quantities.
\begin{defn}[Unique information \citep{bertschinger2014quantifying}]
\label{def:brojaUni} Let $\Delta$ be the set of all joint distributions on $(A, X, Y)$ and $\Delta_P\!=\!\{Q_{AXY}{\in} \Delta$: $Q_{AX}\!=\!P_{AX}$ and $Q_{AY}\!=\!P_{AY}\}$ be the set of joint distributions with same marginals on $(A,X)$ and $(A,Y)$ as the true distribution $P_{AXY}$. Then, $
\uni{A}{X|Y}:=\min _{Q \in \Delta_P} \mutd{Q}{A}{X|Y}. $
Here, $\mutd{Q}{A}{X|Y}$ is the conditional mutual information under joint distribution $Q_{AXY}$ instead of $P_{AXY}$.
\end{defn}

Defining any one of the PID terms is sufficient to derive the others, due to the following relationship:  
\begin{equation}
\mut{A}{X} = \uni{A}{X|Y}+ \rdn{A}{X,Y}.\label{eq:pid_2}
\end{equation}
Intuitively, $\rdn{A}{X,Y}$ can be interpreted as the overlapping portion between $\mut{A}{Y}$ and $\mut{A}{X}$ (see Fig.~\ref{fig:PID}). Hence, $\rdn{A}{X,Y} = \mut{A}{X} - \uni{A}{X |Y}.$ Finally, the synergy corresponds to the remaining information: $
\syn{A}{X,Y} = \mut{A}{X,Y} - \uni{A}{X|Y}  - \uni{A}{Y|X} - \rdn{A}{X,Y}$,  which can be computed once the unique and redundant information terms have been obtained. We also note an additional complementary relationship with conditional mutual information as follows:
\begin{equation}
\mut{A}{X|Y} = \mut{A}{X,Y}-\mut{A}{Y}=\uni{A}{X|Y}+ \syn{A}{X,Y}.\label{eq:pid_3}
\end{equation}

\noindent \textit{Toy Example:} Let $Z{=}(Z_1,Z_2,Z_3)$ with each $Z_i{\sim}$ \iid{} Bern(1/2). Let $X=(Z_1,Z_2,Z_3\oplus N)$, $Y=(Z_2,N)$, and $N\sim $  Bern(1/2) which is independent of $Z$. Here,  $\mathrm{I}(Z;X,Y)=3$ bits. The unique information about $Z$ only in $X$ and not in $Y$ is effectively in $Z_1$. Thus, $\uni{Z}{X|Y} = \mut{Z}{Z_1} = 1$ bit. Redundant information about $Z$ that is in both $X$ and $Y$ is effectively in $Z_2$ and is given by $\rdn{Z}{X,Y}=\mathrm{I}(Z;Z_2)=1$ bit. Synergistic information about $Z$ that is not in either $X$ or $Y$ alone, but is in both of them together is effectively in the tuple $(Z_3\oplus N,N)$, and is given by $\syn{Z}{X,Y} {=} \mut{Z}{(Z_3\oplus N,N)}=1 $ bit. This accounts for the $3$ bits in $\mut{Z}{X,Y}$. 

\section{Proposed Interpretability Framework: \strategy}
\textbf{Notations.} Given an input question $Q$, we denote the LRM-generated reasoning trajectory by $S = \{S_1, S_2, \ldots, S_T\}$, where each $S_i$ represents an intermediate reasoning step. We use $S_{<i}=\{S_1, S_2, .....S_{i-1}\} $ to denote all reasoning steps prior to step $i$. The final answer is denoted by $A$.  We segment the reasoning path into paragraph-level units 
and consider each segment as a distinct reasoning step toward obtaining the final answer $A$, following~\citet{yong2025think}.

\textbf{PID in reasoning trajectories.} The total information of two consecutive reasoning steps $S_{i-1}$ and $S_{i}$ about the final answer $A$ can be decomposed into four non-negative components as follows:
\begin{align} 
\mut{A}{S_{i-1},S_{i}} {=} \rdn{A}{S_{i-1},S_{i}}+ \uni{A}{S_{i-1}\given S_{i}} + \uni{A}{S_{i}\given S_{i-1}} + \syn{A}{S_{i-1},S_{i}}. \label{reasoning_pid}
\end{align}

Each term in \eqref{reasoning_pid} interprets how the reasoning steps contribute toward the final answer $A$. 

Redundant information $\rdn{A}{S_{i-1},S_i}$ captures the overlapping information preserved in both steps. For example, consider an arithmetic problem where the first two reasoning steps are, $S_1$ = $\tfrac{x}{2}$ and $S_2$ = $\tfrac{x}{2} \times \tfrac{3}{5}$, and the final answer $A = x - \tfrac{x}{2} - \tfrac{x}{2} \times \tfrac{3}{5} = \tfrac{x}{5} $. In this case, the term $\tfrac{x}{2}$ computed in $S_1$ is reused in $S_2$ for subsequent computation. As a result, the two steps share overlapping information about the final answer, which results in redundant information. 

Unique information $\uni{A}{S_{i-1}\given S_{i}}$ or $\uni{A}{S_{i}\given S_{i-1}}$ captures the information in one step that cannot be obtained from the other. In the previous example, step $S_2$ involves an additional computation, taking $\frac{3}{5}$ of $\frac{x}{2}$ which is not present in $S_1$ and therefore unique to $S_2$. This implies that $S_2$ provides unique information about the final answer that cannot be obtained from $S_1$. 

Synergistic information $\syn{A}{S_{i-1},S_{i}}$ arises when both steps jointly provide complementary information necessary for determining the final answer, which cannot be obtained from either step alone. For example, consider an arithmetic problem where the first two reasoning steps are $S_1$ = $3x$ and $S_2$ = $2y$, and the final answer $A = 3x+2y$. Here, neither $S_1$ nor $S_2$ is sufficient to determine the final answer $A$. $A$ is obtained only when they are considered together, resulting in synergy.

\textbf{PID under errors.} Intuitively, if a reasoning step contains errors, its task-relevant contribution toward the final answer is diminished. As a result, we expect a reduction in its unique information compared to the previous, non-erroneous case. We formalize this argument in the following theorem.

\begin{thm}[Uniqueness under noise]
\label{thm:uniqueness_noise}

Consider an erroneous step $S_i' = S_i + N$ where N is independent additive noise, $N \perp (A,S_i,S_{i-1})$. Then, $\uni{A}{S_i|S_{i-1}}\geq \uni{A}{S'_i|S_{i-1}}.$

\end{thm}

Theorem~\ref{thm:uniqueness_noise} suggests that introducing synthetic noise into a reasoning step reduces its predictive contribution to the final answer. For example, consider the arithmetic problem in Fig.~\ref{fig:summary} with reasoning steps $S_1$ = $3x$ and $S_2$ = $2y$, and the final answer $A = 3x+2y$. If the calculation in $S_2$ becomes erroneous due to injected noise, $S'_2 = 2y +N$, where N is independent random noise, then $\uni{A}{S_2|S_{1}} \geq \uni{A}{S'_2|S_{1}}$. This indicates that the predictive power of $S_2$ towards $A$ cannot increase under noise. The proof is given in Appendix~\ref{app:theorem1}. 

\textbf{PID to identify repetitive reasoning.} If the current step $S_i$ revisits prior computations in $S_{<i} = \{S_1, S_2, .....S_{i-1}\}$  or presents an alternative formulation to reach the same conclusion, we would consider $S_i$ as repetitive. To identify repetitive reasoning, we propose the following decomposition:

\begin{prop}[Decomposition for repetitive reasoning]
\label{prop:total_info}
The total information of $S_{<i}$ and $S_{i}$ about the final answer $A$ can be decomposed into four non-negative components:
\begin{align*} 
\mut{A}{S_{<i},S_{i}} = \rdn{A}{S_{<i},S_i} +\uni{A}{S_{<i}\given S_{i}} + \uni{A}{S_i\given S_{<i}} + \syn{A}{S_{<i},S_{i}}.
\end{align*}
\end{prop}


When $S_i$ is repetitive, we expect redundant information $\rdn{A}{S_{<i},S_i}$ to be high and dominant over other terms. This leads to our next proposition, a step-wise measure of repetitiveness that contrasts the redundant contribution of \(S_i\) with its unique and synergistic contributions.

\begin{defn}[Step-wise Repetitive Reasoning Index (Step-RRI)]
\label{prop:red}
For some $\eta \geq 1$ and $i > 1$, our proposed measure of repetitiveness in a reasoning step $S_i$ is defined as, \begin{equation} \text{Step-RRI} := \mathrm{RRI}_i =  \rdn{A}{S_i,S_{<i}} - \eta \cdot \max\{\uni{A}{S_i|S_{<i}}, \syn{A}{S_i,S_{<i}}\}.\end{equation} 
\end{defn}

For intuition, consider the simplified case where $S_{<i} {=} S_{i}$. \emph{We will now show how the uniqueness in $S_{i}$ and $S_{<i}$ and synergy become zero, leaving redundancy as the only positive term.} First note that the conditional mutual information $\mut{A}{S_{<i}|S_{i}}=\mut{A}{S_{i}|S_{<i}}=0$ since $\mut{A}{S_{i}|S_{<i}} \overset{\text{(a)}}= H(S_{i}|S_{<i})-H(S_{i}|A,S_{<i}) = 0.$ where (a) is from definition. Then, $\uni{A}{S_{i}|S_{<i}} \overset{(b)}=\mut{A}{S_{i}|S_{<i}} -\syn{A}{S_{i},S_{<i}} \overset{(c)}\leq \mut{A}{S_{i}|S_{<i}} = 0$, where (b) follows from \eqref{eq:pid_3} and (c) is from the non-negativity of PID terms. Similarly,  $\uni{A}{S_{<i}|S_{i}}\leq \mut{A}{S_{<i}|S_{i}} = 0$. Then, $\syn{A}{S_{<i},S_{i}} = \mut{A}{S_{<i}|S_{i}}  - \uni{A}{S_{<i}|S_{i}} $ is also $0$.  Lastly, $\rdn{A}{S_{i},S_{<i}} = \mut{A}{S_{i}} - \uni{A}{S_{i}|S_{<i}} = \mut{A}{S_{i}} = H(A) - H(A|S_{i})$, which is positive as long as there is some dependence between $A$ and $S_{i}$. As a result, Step-RRI will also be positive. We formalize this intuition in the next theorem. The proof (see Appendix~\ref{app:theorem2}) leverages interesting monotonicity properties of the PID measures. In particular, we leverage the fact that $\uni{A}{B|C\cup C'}\leq \uni{A}{B|C}$ (Monotonicity under adversarial side information; see Lemma~\ref{lemma:monotonicity_eve} in Appendix~\ref{app:theory}).



\begin{thm}[Step-RRI under  Repetitive Reasoning]
\label{thm:repititive2}
Let the current step $S_i=(S_A,S_N)$ where $S_A$ is relevant to the answer $A$ and $S_N$ is additional independent text, i.e., $\mut{A}{S_A}>0$ and $S_N \perp A, S_{<i}$ (noise). Suppose the current step is repetitive with the logic already contained in past steps $S_{<i}$, i.e., the relevant part $S_A \subseteq S_{<i}$. Then, we have: (i) $\uni{A}{S_{i}|S_{<i}}$ and $\syn{A}{S_{i},S_{<i}}$ are both zero; (ii) $\text{Step-RRI} = \rdn{A}{S_{i},S_{<i}} - \eta \cdot \max\{\uni{A}{S_i|S_{<i}}, \syn{A}{S_i,S_{<i}}\} >0.$  
\end{thm}

Theorem~\ref{thm:repititive2} highlights that when a reasoning step $S_i$ only repeats information already contained in $S_{<i}$, its unique and synergistic information contributions vanish, and the step becomes purely redundant with respect to the answer, resulting in a positive $\text{Step-RRI}$. 

Next, to characterize reasoning repetitiveness across the entire reasoning trajectory at the individual-problem level, we define \emph{Trajectory-RRI}. Trajectory-RRI quantifies the overall repetitiveness of a reasoning trajectory for a given question by counting the number of reasoning steps identified as repetitive by our measure Step-RRI.

\begin{defn}[Trajectory Repetitive Reasoning Index (Trajectory-RRI)]
\label{def:point-wiseRRI}
For a given question $Q$ with a reasoning trajectory of length $T$, we define the Trajectory-RRI as,
$\text{Trajectory-RRI} := \mathrm{RRI}(Q)
=
\sum_{i=2}^{T}
\mathbb{I}\left[\mathrm{RRI}_i > 0\right],
$
where $\mathbb{I}[\cdot]$ denotes the indicator function.
\end{defn}

\textbf{Implementation of \strategy.} The estimation of PID measures for high-dimensional data with continuous target variables is a non-trivial task. We propose a tractable approach for PID estimation at the step level in reasoning trajectories. Our objective is to obtain stable PID estimates that are both intuitive and interpretable. \strategy consists of several key steps (see Fig.~\ref{fig:framework}) as follows:\begin{figure}[!htbp]
    \centering
    \includegraphics[width=0.97\linewidth]{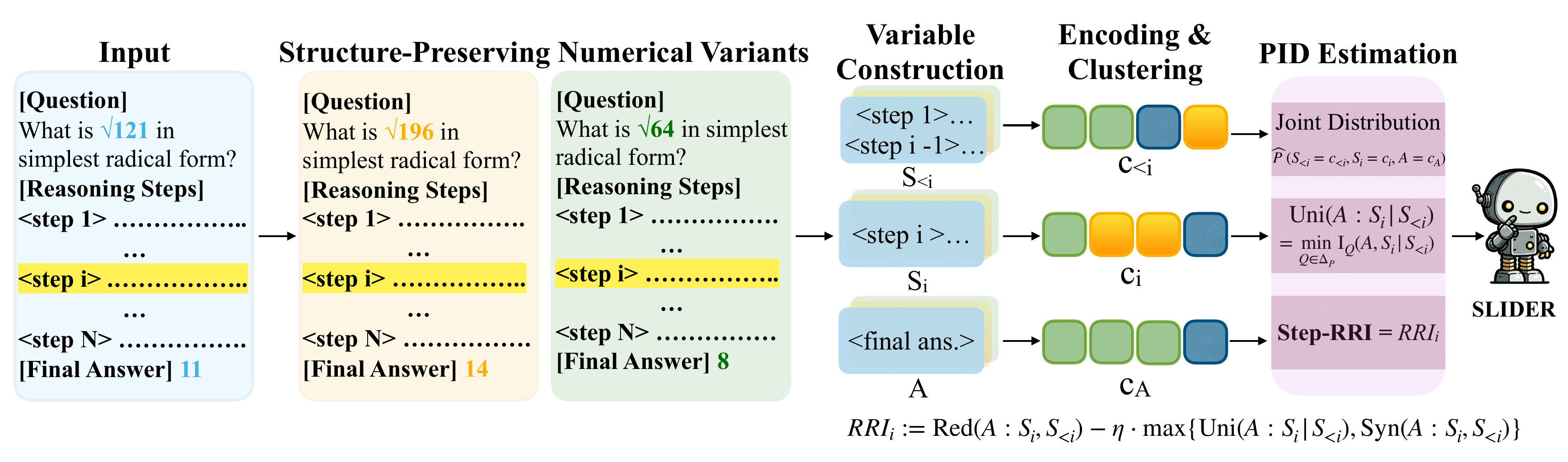}
    \caption{\emph{Our proposed framework \strategy}. 
    Given a question and its reasoning trajectory, the sample generator creates multiple variants by changing only the numerical values while preserving the underlying problem structure. Each generated reasoning trajectory is then segmented into individual reasoning steps using the \texttt{<step i>} tags, from which we construct the preceding reasoning $S_{<i}$ and the current step $S_i$ in text form. Then, we obtain embeddings of the current step $S_i$, preceding steps $S_{<i}$, and final answer $A$ using an encoder-only language model and discretize the resulting embeddings via $k$-means clustering. Finally, we estimate the joint distribution from the discrete representations, compute the PID measures, and Step-RRI.\vspace{-10pt}}
    
    \label{fig:framework}
\end{figure}

\emph{Structure-Preserving Sample Generation.} To estimate PID at the individual reasoning-step level, we require multiple samples to estimate the underlying joint distributions. A single reasoning trajectory provides insufficient samples for reliable estimation, while independently generated trajectories may not match exactly in their reasoning structure and step correspondence. To address this challenge, we construct multiple \emph{structure-preserving numerical variants} of each problem, ensuring exact semantic and structural alignment of corresponding steps across trajectories. These aligned samples enable the estimation of step-level joint distributions and, subsequently, the PID measures and RRI measures.

Specifically, for each original question and its associated reasoning
trajectory, we use GPT-4o-mini~\citep{openai2024gpt4omini} (the prompt template is in
Appendix~\ref{app:data-generation}), unless otherwise stated, to iteratively generate numerical variants. The generator modifies only the numerical values while preserving the question format, reasoning logic, step ordering, and \texttt{<step i>} structure. A generated sample is accepted only if its trajectory contains the same number of reasoning steps as its input, ensuring correspondence between $S_i$ across samples. Each accepted sample is then used as the reference for the next generation step, forming a recursive chain that increases numerical diversity while preserving the underlying reasoning structure. If the generated trajectory does not preserve the step count, the sample is discarded, and generation is reset to the original example. We generate up to $50$ variants per problem and terminate generation after five consecutive structural mismatches. This procedure yields a collection of aligned reasoning trajectories from which we estimate the joint distributions of the step-level representations required for PID.

\emph{Encoding and PID Estimation.} Each generated reasoning trajectory is first segmented into individual reasoning steps using the \texttt{<step i>} tags. For each step \(i\), we define the tagged segment as the current-step variable \(S_i\) and construct the preceding-reasoning variable \(S_{<i}\) by concatenating all steps before \(S_i\) in their original textual form. Next, we use an encoder-only LLM, all-MiniLM-L6-v2~\citep{reimers2019sentence}, unless otherwise stated, to obtain embeddings of the preceding steps \(S_{<i}\), the current step \(S_i\), and the final answer \(A\), respectively. We then discretize these continuous embeddings using k-means clustering with $10$ clusters to obtain categorical variables (refer to Table~\ref{tab:pid_cluster_sensitivity}, Appendix~\ref{app:cluster_sensitivity} for a sensitivity analysis of the number of clusters). Based on the resulting cluster assignments, we estimate the joint distribution, $\widehat{P}(S_{<i}=c_{<i},\,S_i=c_i,\,A=c_A)$, with a slight abuse of notation where we still use $S_i$ to denote the discrete embeddings of the corresponding steps. Then, according to Definition~\ref{def:brojaUni}, we compute the PID measures: redundant information $\rdn{A}{S_{<i},S_i}$, unique information $\uni{A}{S_{<i}\given S_i}$ and $\uni{A}{S_i\given S_{<i}}$, and synergistic information $\syn{A}{S_{<i},S_i}$ from the estimated joint distribution. Specifically, we utilize the CVX estimator proposed by~\citet{liang2023multimodal} for PID estimation. Finally, following Definition~\ref{prop:red} and~\ref{def:point-wiseRRI}, we compute Step-RRI and Trajectory-RRI, respectively. For further implementation details, refer to Appendix~\ref{app:data-generation}.

\section{Experimental Results}
\label{sec:experiments}
In this section, we evaluate \strategy across synthetic arithmetic problems, PRMBench~\citep{song2025prmbench}, GSM8K~\citep{cobbe2021gsm8k}, and AIME2025~\citep{aime25}. We examine the effect of injected errors on PID components, evaluate Step-RRI for detecting step-level ground-truth redundancy against cosine similarity and InfoGain~\citep{yong2025think}, and assess the relationship between Trajectory-RRI and reasoning length across reasoning and non-reasoning models. Finally, we use Trajectory-RRI to guide data selection for fine-tuning Qwen2.5-7B-Instruct~\citep{qwen2.5} toward more efficient reasoning. Additional dataset details are provided in Appendix~\ref{app:experiments}.



\textbf{\strategy uncovers interpretable information dynamics under synthetic noise.} We consider the following problem: \emph{Arithmetic problem:} 
``x = \{x\}, y = \{y\}. Please calculate the following: 
1. $3x$, 
2. $2y$, 
3. $3x + 2y$.'' 
We sample two integers $x$ and $y$ independently and uniformly from $[1, 10^5)$. We generate $500$ samples in exact format, each containing exactly three reasoning steps. The intermediate steps are $S_1 = {3x}$, $S_2 = {2y}$, and $S_3 = {3x+2y}$, and the final answer $A = 3x+2y$ (see Fig.~\ref{fig:cot}(d)). To introduce errors, we perturb an individual step with probability
$p$ by adding an integer sampled uniformly from $[-1000,10^6]$ to its
value, while keeping the remaining steps and the final answer unchanged.
We vary $p$ to analyze how increasing the probability of errors
affects the PID components.\begin{figure}[!htbp]
        \centering
        \includegraphics[width=\linewidth]{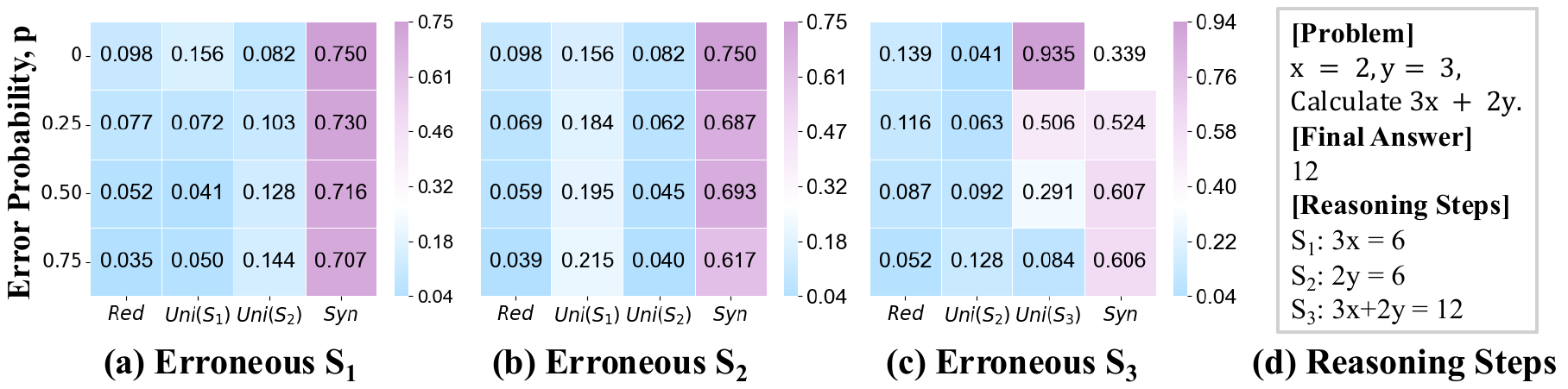}
    \caption{\emph{Interpretable trends of PID measures under errors for synthetic arithmetic problems.} (a,b) Increasing error probability \(p\) reduces synergy and the unique information of erroneous steps, while increasing the unique contribution of correct steps. (c) Perturbing \(S_3\) similarly decreases \(\mathrm{Uni}(S_3)\) as \(p\) increases. (d) Reference example for the reasoning steps.
    \label{fig:cot} } 
    \vspace{-7pt}
    \end{figure}

Fig.~\ref{fig:cot} illustrates how the information contributed by intermediate reasoning steps toward the final answer evolves under synthetic noise. We first consider $S_1=3x$ and $S_2=2y$ in Fig.~\ref{fig:cot}(a,b). For this arithmetic problem, individually, neither $S_1=3x$ nor $S_2=2y$ is sufficient to determine the target $A$. However, when $S_1$ and $S_2$ are considered jointly, their combined information fully specifies $A$ (considering no error in the reasoning steps). This implies that the information about $A$ is predominantly synergistic/complementary across $S_1$ and $S_2$. As the error probability increases in either $S_1$ or $S_2$, synergy decreases because the error limits the inference of the final answer, even when considering both steps together. At the same time, the unique information of the erroneous step decreases, while that of the corresponding correct step increases, as the correct step carries answer-relevant information that is increasingly unavailable from its corrupted counterpart. In Fig.~\ref{fig:cot}(c), we consider the same problem with the last two reasoning steps, $S_2 = 2y$ and $S_3 = 3x + 2y$, and the final answer $A = 3x + 2y$. When $S_3$ is correct, it directly determines the final answer $A$ and therefore contributes substantial unique information relative to $S_2=2y$. However, if there is an error in the calculation in $S_3$, its uniqueness decreases, since some examples can no longer provide useful information about $A$ that is not already available from $S_2$. We can also use \strategy to identify incorrectness for the GSM8K word problem~\citep{cobbe2021gsm8k} (see details in Appendix~\ref{app:noise}).

\textbf{\strategy identifies step-wise repetitiveness in reasoning trajectories.}\begin{figure}[!htbp]
    \centering
        \includegraphics[width=\linewidth]{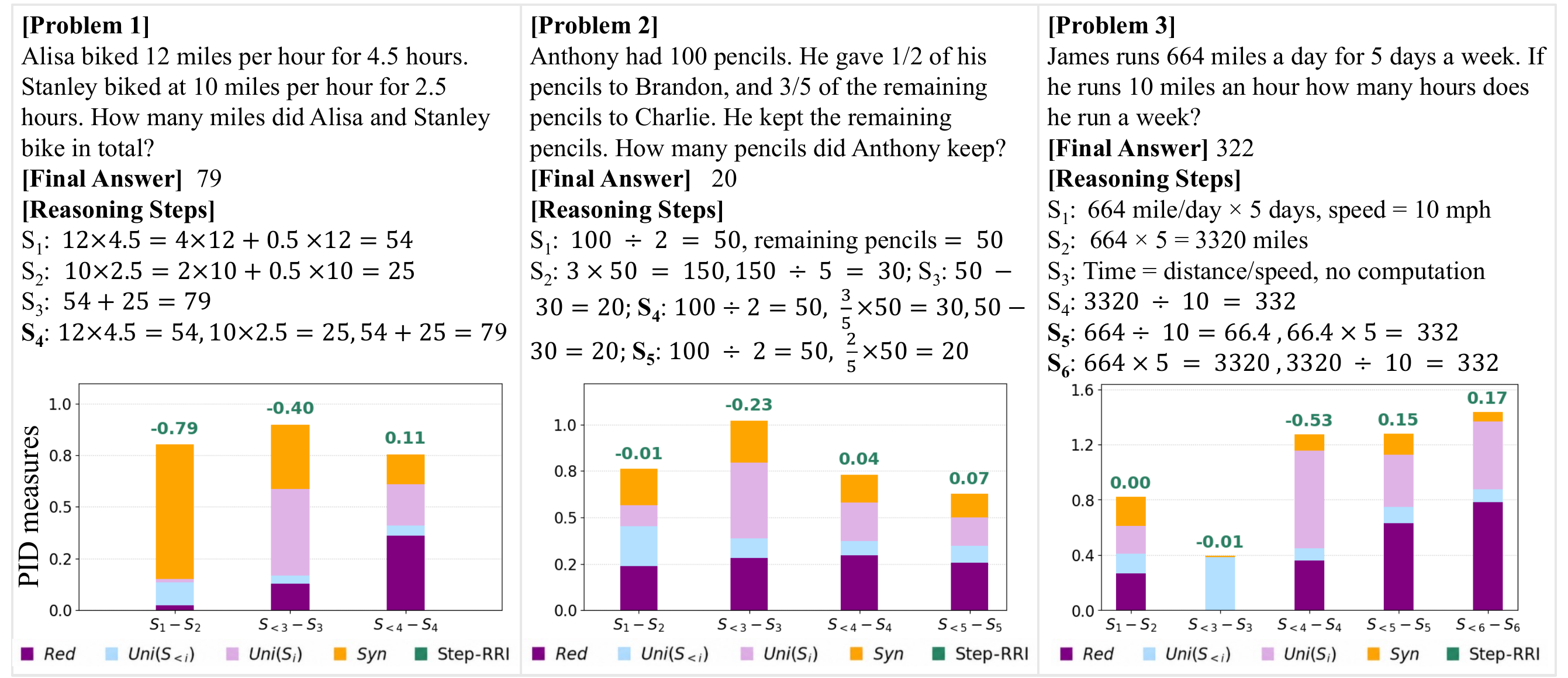}
    \caption{\emph{PID estimation across reasoning trajectories.}
    Redundant information becomes dominant compared to uniqueness in $S_i$ and synergistic information when the current step $S_i$ revisits answer-relevant information already established in the preceding steps $S_{<i}$, resulting in positive Step-RRI. 
}
\vspace{-10pt}
    \label{fig:qwq}
\end{figure} To assess whether Step-RRI aligns with the actual repetitiveness between reasoning steps, we analyze reasoning trajectories generated by QwQ-32B on three GSM8K word problems. Implementation details and additional experiments using reasoning trajectories generated by DeepSeek-R1-Distill-Qwen-32B (see Fig.~\ref{fig:r1_pid}) are provided in Appendix~\ref{app:contribution2}. Fig.~\ref{fig:qwq} presents the corresponding PID and Step-RRI estimates, illustrating how answer-relevant information evolves across reasoning steps and how repetitive steps emerge within each trajectory based on Step-RRI. We observe that for Problem 1, synergistic information is initially dominant, indicating that the early reasoning steps provide complementary information required to reach the final answer. As the reasoning progresses toward the answer, the unique information of the current step increases. Specifically, $\mathrm{Uni}(S_3)$ is maximized at $S_{<3}$--$S_3$ for Problems 1 and 2, while $\mathrm{Uni}(S_4)$ is maximized at $S_{<4}$--$S_4$ for Problem 3. This suggests that this reasoning step $S_i$ contributes information that is largely unavailable from previous steps, as it directly contains the final answer.
\strategy further localizes repetitive reasoning to individual steps using Step-RRI.  The redundant contribution dominates the unique contribution of $S_i$ and its synergistic contribution with $S_{<i}$ when the model either repeats earlier computations or performs alternative re-computations that recover answer-relevant information already contained in $S_{<i}$. As a result, Step-RRI becomes positive.\begin{wrapfigure}[13]{r}{0.3\columnwidth}  
  \vspace{-10pt}
  \centering
  \includegraphics[width=\linewidth]{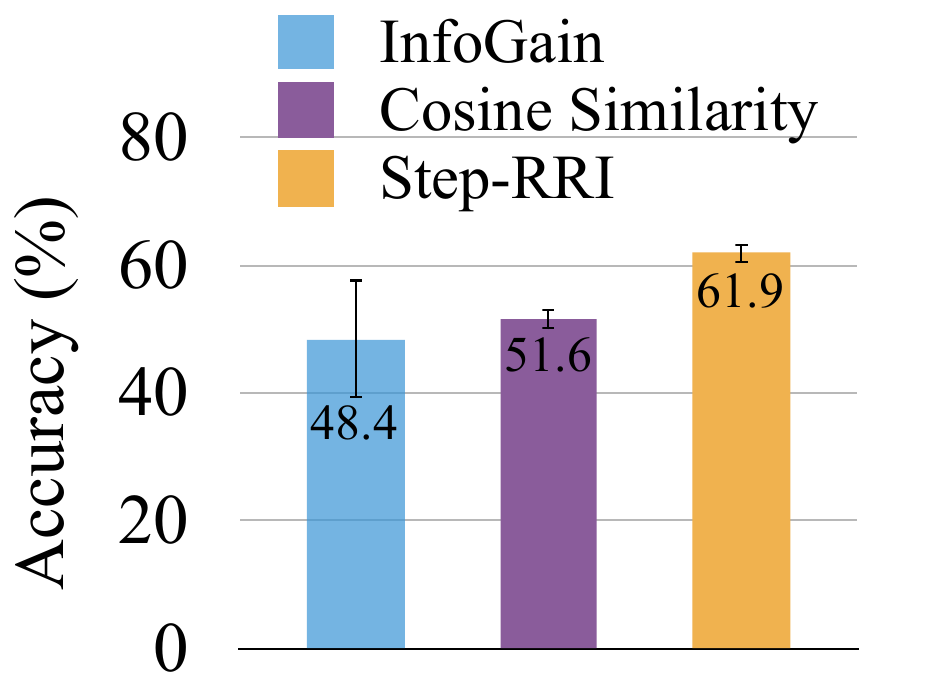}
  \caption{\emph{Step-level redundancy detection.} Step-RRI achieves higher accuracy than baselines on PRMBench.}
  \label{fig:prmbench}
\end{wrapfigure}   Accordingly, with $\eta=1.25$, Step-RRI becomes positive at $S_4$ in Problem~1, $S_4$ and $S_5$ in Problem~2, and $S_5$ and $S_6$ in Problem~3. Thus, Step-RRI enables \strategy to localize repetitive reasoning to specific steps within each problem.

Next, we quantitatively evaluate the effectiveness of Step-RRI for identifying true repetitive reasoning steps using PRMBench~\citep{song2025prmbench} since this dataset has annotations for step-level redundancy. We focus specifically on the redundancy error class and compare Step-RRI against two baselines: InfoGain~\citep{yong2025think} and Cosine Similarity. For each method, we calibrate the decision threshold on a held-out set by training a linear classifier and use the resulting threshold for step-level ground-truth redundancy detection on the test set. Fig.~\ref{fig:prmbench} reports the detection performance. Step-RRI achieves the highest accuracy of $61.91\%$, improving over InfoGain and Cosine Similarity by $13.53$ and $10.31$ percentage points. More details on the baselines are in Appendix~\ref{app:contribution2}.


\textbf{\strategy reveals a strong correlation between Trajectory-RRI and reasoning length.}
To demonstrate the practical utility of Trajectory-RRI, we examine its relationship with actual reasoning length across multiple models. We consider the QwQ-32B~\citep{team2025qwq}, DeepSeek-R1-Distill-Qwen-32B~\citep{guo2025deepseek}, and and GPT-4.1~\citep{openai2025gpt41}. For QwQ-32B and DeepSeek-R1-Distill-Qwen-32B, each response contains a thinking trajectory followed by the final response, separated by the \texttt{</think>} tag; we use the text preceding this tag as the reasoning trajectory. For GPT-4.1, we consider the entire generated response as the reasoning trajectory. See Appendix~\ref{app:contribution3} for additional details on the experimental setup.

We evaluate on GSM8K~\citep{cobbe2021gsm8k} and AIME2025~\citep{aime25}. For GSM8K, we consider 1,000 randomly selected problems with correct final answers, evaluating QwQ-32B, DeepSeek-R1-Distill-Qwen-32B, and GPT-4.1.
\begin{table}
\centering
\small
\setlength{\tabcolsep}{3pt}

\caption{Spearman correlation between average Trajectory-RRI and reasoning length.}
\label{tab:rri_correlation}
\begin{tabular}{lccccccc}
\toprule
& \multicolumn{2}{c}{DeepSeek-R1}
& \multicolumn{2}{c}{QwQ-32B}
& \multicolumn{2}{c}{GPT-4.1} \\
\cmidrule(lr){2-3}
\cmidrule(lr){4-5}
\cmidrule(lr){6-7}
Dataset
& $\rho$ & $p$-val
& $\rho$ & $p$-val
& $\rho$ & $p$-val \\
\midrule
GSM8K
& 0.95 & 0.00
& 0.91 & 0.00
& 0.73 & 0.00 \\
AIME2025
& 0.89 & 0.01
& 0.75 & 0.01
& 0.82 & 0.03 \\
\bottomrule
\end{tabular}
\end{table} 
For each trajectory, we compute Trajectory-RRI using \strategy and measure reasoning length as the number of tokens in the corresponding reasoning trajectory. To mitigate variability from problem-specific step segmentation, we group trajectories by reasoning length and report the average Trajectory-RRI and reasoning length within each bin. As shown in Fig.~\ref{fig:rri_tokens} and Table~\ref{tab:rri_correlation}, average Trajectory-RRI is positively associated with reasoning length across the evaluated settings. On GSM8K, the Spearman correlations (see Appendix~\ref{app:experiments} for its definition) are $\rho=0.95$ for DeepSeek-R1-Distill-Qwen-32B, $\rho=0.91$ for QwQ-32B, and $\rho=0.73$ for GPT-4.1. On AIME2025, DeepSeek-R1-Distill-Qwen-32B, QwQ-32B, and GPT-4.1, yield Spearman correlations of $\rho=0.89$, $0.75$, and $0.82$, respectively. The corresponding $p$-values reported in Table~\ref{tab:rri_correlation} are also small, indicating that these correlations are statistically significant.
\emph{The consistently strong positive correlations across these models indicate that trajectories identified as more repetitive by our measure also tend to require more reasoning tokens.}

\begin{figure*}[t]
    \centering

    \begin{minipage}[t]{0.49\textwidth}
        \centering
        \includegraphics[width=\linewidth]{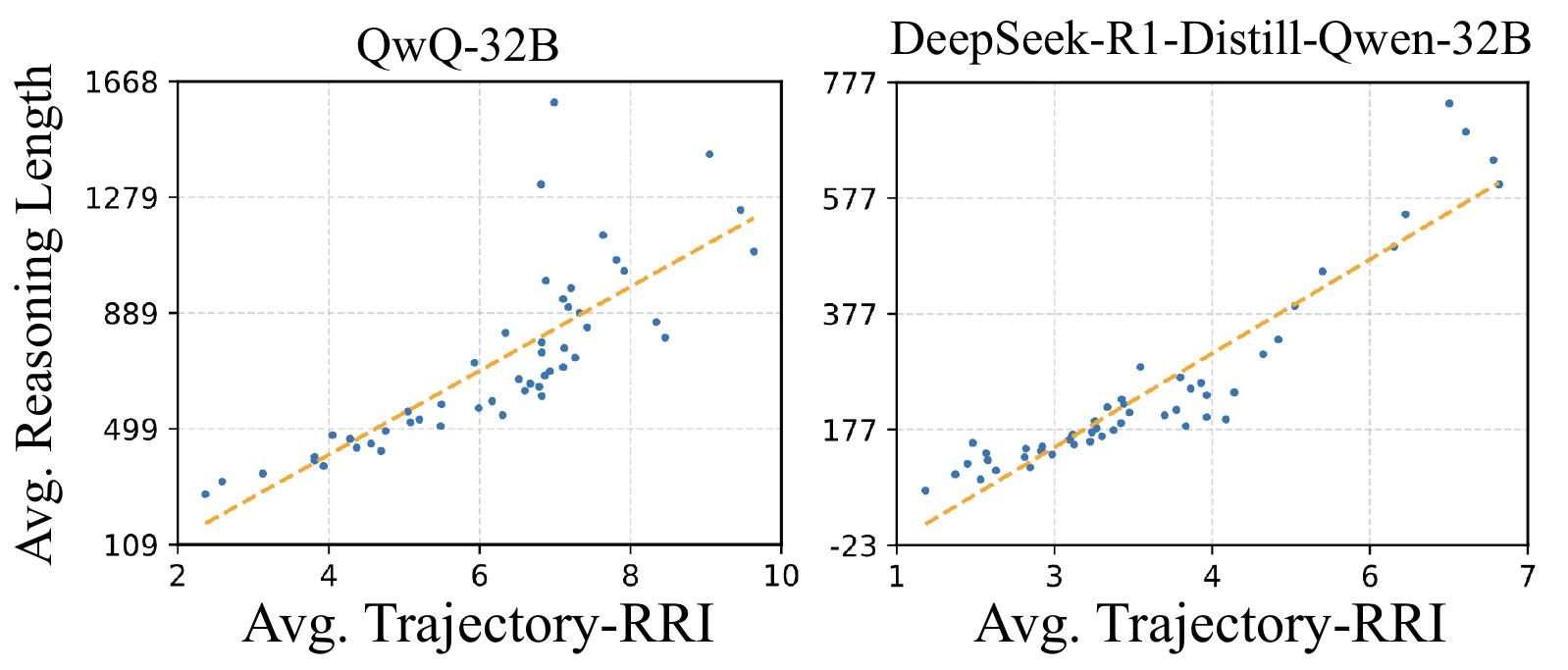}
        \captionof{figure}{
        \emph{Trajectory-RRI strongly correlates with reasoning length.} Avg. Trajectory-RRI exhibits a positive monotonic relationship with average thinking tokens for both models (left; $\rho=0.91$, right; $\rho=0.95$) on GSM8K, indicating higher RRI is associated with longer trajectories.
        }
        \label{fig:rri_tokens}
    \end{minipage}
    \hfill
    \begin{minipage}[t]{0.49\textwidth}
        \centering
        \includegraphics[width=\linewidth]{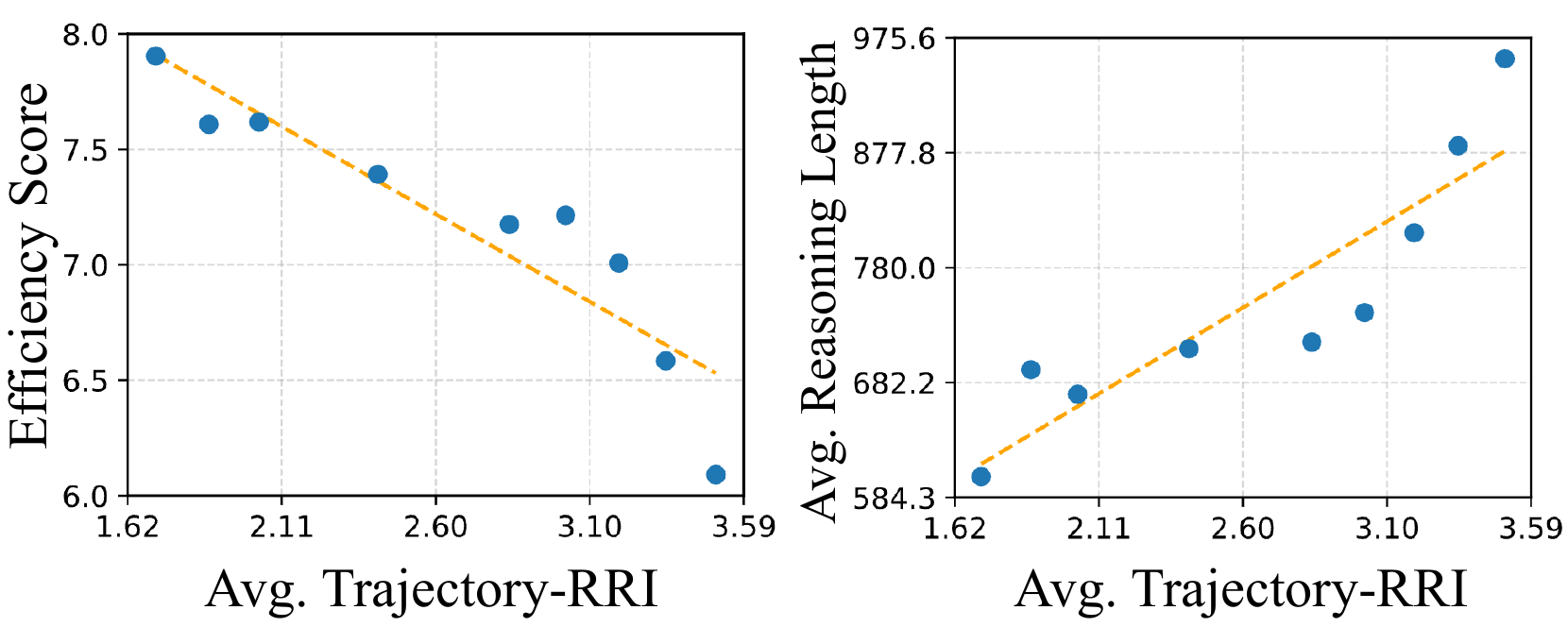}
        \captionof{figure}{
        \emph{Trajectory-RRI guided fine-tuning improves reasoning efficiency.} As the average Trajectory-RRI of the training data increases, the efficiency score of the fine-tuned model on GSM8K decreases (left; $\rho=-0.97$), while its average tokens increase (right; $\rho=0.98$).
        }
        \label{fig:NetRRI}
    \end{minipage}
\vspace{-10pt}
\end{figure*}

\textbf{Trajectory-RRI serves as a signal for SFT data selection.} Motivated by the strong positive correlation between average Trajectory-RRI and reasoning length, we investigate whether average Trajectory-RRI can be used to guide data selection for supervised fine-tuning toward more efficient reasoning while preserving task performance.
We use 1,000 GSM8K problems, along with their corresponding reasoning trajectories generated by QwQ-32B and DeepSeek-R1-Distill-Qwen-32B, yielding two candidate reasoning trajectories for each problem. Using \strategy, we estimate the Trajectory-RRI of each trajectory and construct training sets with varying average Trajectory-RRI by selecting, for each problem, the trajectory with the higher Trajectory- RRI with probability $p$. We then fine-tune Qwen2.5-7B-Instruct~\citep{qwen2.5} on each dataset using Low-Rank Adaptation (LoRA)~\cite{hu2021lora} and evaluate the resulting models on the GSM8K test set in terms of accuracy, reasoning length, and an LLM-as-a-judge efficiency score. As shown in Fig.~\ref{fig:NetRRI}, the average Trajectory-RRI of the training data exhibits a strong negative correlation with an LLM-as-a-judge efficiency score of the fine-tuned models (Spearman's correlation $\rho=-0.97$, $p-\text{val}=2.2\times10^{-5}$) and a strong positive correlation with their actual reasoning length, measured by the total number of tokens generated during inference ($\rho=0.98$, $p-\text{val}=1.9\times10^{-6}$), while test accuracy remains within a comparable range (see Table~\ref{tab:rri_accuracy} in Appendix~\ref{app:NetRRI}). \emph{These results demonstrate that Trajectory-RRI can serve as an actionable signal for curating fine-tuning data: training on trajectories with lower Trajectory-RRI leads to more concise and efficient reasoning without a substantial degradation in task performance.} See Appendix~\ref{app:NetRRI} and~\ref{app:runtime} for additional implementation and computational details, respectively.  

\textbf{Conclusion.} This work introduces \strategy, a novel information-theoretic framework for interpreting reasoning trajectories using Partial Information Decomposition. We propose \emph{Step-RRI} to identify repetitive reasoning at the individual-step level and \emph{Trajectory-RRI} to quantify overall repetitiveness within a reasoning trajectory for a specific problem. Our results demonstrate that RRI measures effectively identify true repetitive reasoning, strongly correlate with actual reasoning length across LLMs, and can guide data selection for fine-tuning toward more efficient reasoning while largely preserving task performance. We discuss limitations and future work in Appendix~\ref{app:limitations}.

\bibliography{iclr2027_conference}
\bibliographystyle{iclr2027_conference}
\newpage
\appendix
\startcontents[appendix]
\begin{center}
    {\LARGE\bfseries Appendix}
\end{center}
\vspace{1em}

\printcontents[appendix]{}{1}{}

\clearpage

\section{Limitations and Broader Impacts}
\label{app:limitations}
\paragraph{Limitations.}
Our framework estimates PID from encoded and discretized representations of reasoning steps. This process may introduce information loss and make the estimated PID components sensitive to the choice of encoder, clustering procedure, and other estimation parameters. Future work could investigate continuous PID estimators and alternative representations to reduce this dependence. In addition, our structure-preserving sample generation relies on an LLM to generate aligned numerical variants, and the quality of the resulting PID estimates depends on how well these variants preserve the underlying reasoning structure. Evaluating step-level incorrectness under synthetic noise also requires knowledge of the expected problem structure or an approximate correspondence to correct reasoning steps, which may limit its applicability to domains where such structure is difficult to obtain. Finally, our experiments primarily focus on mathematical reasoning tasks. Evaluating \strategy and RRI measures across broader reasoning domains and model families remains an important direction for future work.

\paragraph{Broader Impacts.}
Quantifying how individual steps contribute to a model's reasoning has broader implications for understanding and improving large language models. By characterizing repetitive reasoning at both the individual-step and trajectory levels, Step-RRI and Trajectory-RRI can support more interpretable auditing of reasoning behavior and provide signals for selecting more efficient reasoning trajectories during post-training. In particular, RRI measures have the potential to serve as theoretically grounded metrics for data selection when the goal is to promote efficient reasoning. While commonly used proxies such as reasoning length provide a coarse measure for curating efficient reasoning data, RRI directly characterizes repetitiveness based on the information contributed by reasoning steps, providing a principled basis for selecting less repetitive trajectories. More broadly, \strategy provides an information-theoretic perspective for studying when reasoning steps contribute new, shared, or complementary answer-relevant information. Such tools may help researchers diagnose inefficient or unreliable reasoning and develop models that produce more concise and interpretable reasoning trajectories.

\section{Appendix to Theoretical Results}
\label{app:theory}
First, we include two results on PID from \citep{banerjee2018unique}.

\begin{lem}[Monotonicity under local operations on $B$~\citep{banerjee2018unique}] Let $B'=f(B)$ where $f(\cdot)$ is a deterministic function. Then,
$$\uni{A}{B|C} \geq \uni{A}{B'|C}   .$$
\label{lemma:monotonicity_bob}
\end{lem}
This result is derived in \citep[Lemma 31]{banerjee2018unique}. 

\begin{lem}[Monotonicity under adversarial side information~\citep{banerjee2018unique}]
For all $(A,B,C,C')$, we have:
$$\uni{A}{B|C \cup C'} \leq \uni{A}{B|C}.$$
\label{lemma:monotonicity_eve}
\end{lem} 
This result is derived in  \citep[Lemma 32]{banerjee2018unique}. 

\subsection{Proof to Theorem~\ref{thm:uniqueness_noise}}
\label{app:theorem1}



\begin{proof}[Proof of Theorem~\ref{thm:uniqueness_noise}] 
The proof follows from monotonicity under local operations, Lemma~\ref{lemma:monotonicity_bob}, since $S_i' = S_i+N$ is a local operation on $S_i$. 

We provide the detailed proof here for completeness. Since $N \perp (A,S_i,S_{i-1})$, the random variables $(A,S_{i-1})$, $S_i$, and $S_i'$ form a Markov chain, $(A,S_{i-1})- S_i - S_i'$. Let $P'$ be the marginal distribution of $(S_i',A,S_{i-1})$, and $P$ be the marginal distribution of $(S_i,A,S_{i-1})$. Let $Q^*
=\arg\min_{Q \in \Delta_P} I_Q(S_i;A\mid S_{i-1}),$
and let
\[
Q^{*'}(s_i',a,s_{i-1})
=
\sum_{s_i} P'(s_i'\mid s_i)Q^*(s_i,a,s_{i-1}).
\]
Then $Q^{*'} \in \Delta_{P'}$. By definition,
\begin{align*}
\uni{S_i}{A|S_{i-1}}
&= I_{Q^*}(S_i;A\mid S_{i-1}) \\
&\overset{(a)}\geq I_{Q^{*'}}(S_i';A\mid S_{i-1}) \\
&\geq \min_{Q' \in \Delta_{P'}}
    I_{Q'}(S_i';A\mid S_{i-1}) \\
&= \uni{S_i'}{A|S_{i-1}},
\end{align*}
Here, (a) uses the conditional form of the data processing inequality. Therefore, $\uni{S_i'}{A\mid S_{i-1}}\leq\uni{S_i}{A\mid S_{i-1}}$, i.e., the unique information contributed by the noisy step $S_i'$ about $A$ cannot exceed that of the original step $S_i$.

\end{proof}

\subsection{Proof to Theorem~\ref{thm:repititive2}}
\label{app:theorem2}

\begin{proof}[Proof of Theorem~\ref{thm:repititive2}]
To begin with, 
\begin{align*}
\uni{A}{S_{i}|S_{<i}} & \overset{(a)}{=} \mut{A}{S_{i}|S_{<i}} -\syn{A}{S_{i},S_{<i}} \\
& \overset{(b)}{\leq} \mut{A}{S_{i}|S_{<i}} \\
& \overset{(c)}= \mut{A}{S_A,S_N|S_{<i}} \\
& \overset{(d)}= \mut{A}{S_N|S_{<i}} + \mut{A}{S_A|S_{<i},S_N} \\
& \overset{(e)}= 0 + \mut{A}{S_A|S_{<i},S_N} \\
& \overset{(f)}=0.
\end{align*}

Here, step (a) holds from  \eqref{eq:pid_3} and step (b) holds from the non-negativity of PID terms, $\syn{A}{S_{i},S_{<i}}\geq 0 $. Step (c) holds because $S_i=(S_A,S_N)$. Step (d) follows directly from the chain rule of mutual information. Step (e) holds because $S_N \perp A, S_{<i}$ and step (f) holds because $S_A \subseteq S_{<i}$.

Similarly, $\syn{A}{S_{i},S_{<i}} = \mut{A}{S_{i}|S_{<i}} - \uni{A}{S_{i}|S_{<i}} \leq \mut{A}{S_{i}|S_{<i}} = 0$.

Next, since $\text{Step-RRI} := \rdn{A}{S_i,S_{<i}} - \eta \cdot \max\{\uni{A}{S_i|S_{<i}}, \syn{A}{S_i,S_{<i}}\}$, and $\max\{\uni{A}{S_i|S_{<i}}, \syn{A}{S_i,S_{<i}}\} = 0,$ we have,
\begin{align*}
\text{Step-RRI} = \rdn{A}{S_{i},S_{<i}} & \overset{(a)}{=} \mut{A}{S_{<i}} -  \uni{A}{S_{<i}|S_{i}} \\
& \overset{(b)}{\geq} \mut{A}{S_{<i}} -  \uni{A}{S_{<i}|S_A} \\
& \overset{(c)}= \rdn{A}{S_A,S_{<i}} \\
& \overset{(d)}= \mut{A}{S_A} - \uni{A}{S_A|S_{<i}} \\
& \overset{(e)}\geq  \mut{A}{S_A} - \mut{A}{S_A|S_{<i}} \\
& \overset{(f)}= \mut{A}{S_A} \\
& \overset{(g)}> 0.
\end{align*}
Here, step (a) holds from \eqref{eq:pid_2}, step (b) holds from Lemma~\ref{lemma:monotonicity_eve} (Monotonicity under adversarial side information) for unique information since $S_i=(S_A,S_N)$, step (c) and (d) hold from \eqref{eq:pid_2}, step (e) holds from \eqref{def:brojaUni}, $\uni{A}{S_{<i}|S_A} \leq \mut{A}{S_A|S_{<i}}$ and non-negativity of PID terms, step (f) holds because $S_A \subseteq S_{<i}$, and step (g) holds since the mutual information of the relevant part $S_A$ is positive. 

\end{proof}

\section{Appendix to Implementation of \strategy}
\subsection{Structure-Preserving Sample Generation}
\label{app:data-generation}
A key component of \strategy is the \emph{Structure-Preserving Sample Generation} procedure, which constructs multiple samples while preserving the underlying reasoning structure. For responses generated by QwQ-32B~\citep{team2025qwq} and DeepSeek-R1-Distill-Qwen-32B~\citep{guo2025deepseek} on \emph{GSM8K}~\citep{cobbe2021gsm8k} and \emph{AIME2025}~\citep{aime25}, before sample generation, we segment each reasoning trajectory and insert explicit \texttt{<step i>} tags to facilitate consistent step tracking across generated variants. We first segment the reasoning trajectory at double-line breaks when the resulting number of steps is fewer than 11. For trajectories with more than 10 steps, we instead use Qwen2.5-72B-Instruct~\citep{qwen2.5} as an automated reasoning-trace segmenter to produce more balanced and logically coherent steps. This prevents highly uneven segmentation, where an individual current step may be substantially shorter than the accumulated preceding reasoning. Because both are encoded into representations of the same dimensionality, such an imbalance may cause the much longer preceding reasoning to lose potentially relevant information during encoding. The automated segmenter therefore aims to produce more semantically coherent reasoning steps while preserving the underlying reasoning structure. We use the following prompt for this segmentation:

\begin{myquote}
\textbf{[System Prompt]}

You are an expert reasoning-trace segmenter. Your only task is to split a raw reasoning trace into sequential logical segments or steps by inserting markup tags.
\vspace{0.5em}
CRITICAL REGULATION FOR CONTENT PRESERVATION:
\begin{itemize}[itemsep=1pt]
\item Do NOT rewrite, paraphrase, summarize, or alter a single word of the original trace.
\item Do NOT compress or skip text. Every sentence, equation, repetition, error, and alternative reasoning method must remain word-for-word exactly as it appeared originally.
\item Keep all redundant loops, repetitive checks, and alternate thinking trajectories completely intact.
\end{itemize}

Formatting Rules:
\begin{itemize}[itemsep=1pt]
\item Output the full, unchanged text prefixed strictly by step markers like this:

<step1>: [First block of original text, entirely unchanged]

<step2>: [Second block of original text, entirely unchanged]

...

<stepN>: [Final block of original text, entirely unchanged]

\item Maximum number of steps: 10

\end{itemize}
Output must strictly contain ONLY the segmented text following this exact format without any conversational intro or markdown fences.

\textbf{[User Prompt]}
\vspace{0.5em}

QUESTION: \{question\}

REASONING TRACE: \{trace\}

Task:
Segment the reasoning trace into logical steps using the required format. Do NOT modify, delete, or rephrase any text.
Return ONLY the segmented trace.

\end{myquote}

For the redundancy class of the \emph{PRMBench}~\citep{song2025prmbench} dataset, the reasoning steps are already segmented. We prepend \texttt{<step i>} tags to each step to format the input text for sample generation. Next, we generate 50 numerical variants of each problem and its corresponding reasoning trajectory using GPT-4o-mini~\citep{openai2024gpt4omini} for GSM8K and PRMBench, and GPT-5-mini~\citep{openai2025gpt5mini} for AIME2025 to accommodate its more challenging problems. The prompts require each variant to preserve the original reasoning structure and the inserted \texttt{<step i>} tags. This allows corresponding reasoning steps to remain aligned across generated samples, facilitating subsequent step-level information estimation. We use the following prompt for structure-preserving sample generation for GSM8K and PRMBench:

\clearpage

\begin{myquote}
\textbf{[System Prompt]}
\vspace{0.5em}

You are an expert synthetic math dataset generator.

Rules:

\begin{enumerate}[label=\arabic*.,leftmargin=15pt,itemsep=2pt,topsep=5pt]
    \item Change ONLY numerical values in the problem. Do NOT modify any other text.
    \item Preserve the Chain-of-Thought structure EXACTLY:
    \begin{itemize}[leftmargin=15pt,itemsep=1pt]
        \item  Same number of steps where steps are given as: <step1>.... <step2> .... <stepN>....
        \item Same step order
        \item Same reasoning logic per step (even if repetitive, redundant, or unnecessary)
        \item Preserve step formatting exactly as written in the input, including:
        <step1>.... <stepN>....
    \end{itemize}
    \item Return valid JSON only.
\end{enumerate}

\vspace{0.8em}

\textbf{[User Prompt]}
\vspace{0.5em}

Given the reference data input:

\vspace{0.5em}

Problem:

\texttt{\{Original Question\}}

\vspace{0.5em}

Solution:

\texttt{\{Original Reasoning Trace\}}

\vspace{0.8em}

Generate exactly 1 numerical variation of the reference problem and its corresponding solution.
Return exactly this JSON structure:

\begin{lstlisting}[basicstyle=\ttfamily\small]
{
  "problem1": "<Insert problem variant 1>",
  "solution1": "<Insert full solution variant 1>",
  "answer1": "<Insert final numerical answer as number>"
}
\end{lstlisting}
Return VALID JSON ONLY.

\end{myquote}

Given the greater difficulty of AIME2025 problems, we adopt a two-stage generation process. First, we use the following prompt to generate valid numerical variants of the original problem.

\begin{myquote}
\textbf{[System Prompt]}
\vspace{0.5em}

Given a math problem, generate ONE numerical variant by changing only numerical values while preserving the original wording and problem structure.

Solve the generated problem carefully. If it is valid and solvable, end with boxed\{your answer\}. If it is invalid or unsolvable, end with boxed\{no solution\}.

Output:

<variant>
[generated new problem]
</variant>

<solution>
[solution of the new problem]
</solution>

You must end your response with boxed\{your answer\} every time!

\vspace{0.8em}

\textbf{[User Prompt]}
\vspace{0.5em}

Original problem:
 {question}
Remember to box your final answer via boxed\{your answer\}.

\end{myquote}

Then, using the generated solvable problems, we use the following prompt to generate structure-preserving samples.

\begin{myquote}
\textbf{[System Prompt]}
\vspace{0.5em}

You are given (1) a reference math problem, (2) its reference reasoning trajectory,
and (3) a numerical variant of that problem.

Solve the numerical variant by adapting the REFERENCE trajectory as conservatively as possible.

Requirements:

1. Change only numerical values, calculations, symbols, or words necessary for the changed numerical values.

2. Preserve the Chain-of-Thought structure EXACTLY:

   - Same number of steps: <step1>... <step2>... <stepN>...
   - Same step order
   
   - Same reasoning logic per step
   
   - Same step formatting
   
3. Solve the numerical variant correctly.

4. Return only the aligned trajectory followed by boxed\{your answer\}.

\vspace{0.8em}

\textbf{[User Prompt]}
\vspace{0.5em}

REFERENCE PROBLEM:
<question>

REFERENCE TRAJECTORY:
<thinking\_trajectory>

NUMERICAL VARIANT TO SOLVE:
<variant\_problem>
\end{myquote}


\subsection{Sensitivity Analysis}
\label{app:cluster_sensitivity}
A key component of \strategy is the discretization step, for which we employ k-means clustering on the embeddings. To evaluate how sensitive the PID estimates are to the choice of cluster count, we compute PID values for Problem 3 in Fig.~\ref{fig:qwq} using 10, 15, and 20 clusters (Table~\ref{tab:pid_cluster_sensitivity}). The results show that increasing the number of clusters generally yields higher PID values, though the magnitude of this increase differs across components. Crucially, the overall trends remain stable: the patterns of redundancy, uniqueness, and synergy are largely preserved across different cluster choices. For example, Table~\ref{tab:pid_cluster_sensitivity} shows that in the row $S_{<4}-S_{4}$, uniqueness in $S_4$ consistently dominates across all cluster counts $n_{\text{cluster}}$. In contrast, for $S_{<5}-S_{5}$ and $S_{<6}-S_{6}$, redundant information remains the dominant component regardless of the number of clusters $n_{\text{cluster}}$.

\definecolor{headerblue}{RGB}{218,239,249}
\definecolor{rowgray}{RGB}{248,248,248}
\begin{table}[!htbp]
\centering
\caption{Sensitivity of PID components for $S_{<i}$ and $S_i$ to the number of clusters.}
\label{tab:pid_cluster_sensitivity}

\small
\renewcommand{\arraystretch}{1.08}
\setlength{\tabcolsep}{5pt}

\begin{tabular}{ccccc}
\toprule

\textbf{\#Clusters} &
\textbf{Redundancy} &
\textbf{Unique ($S_{<i}$)} &
\textbf{Unique ($S_i$)} &
\textbf{Synergy}
\\


\midrule

\rowcolor{headerblue}
\multicolumn{5}{c}{$S_{<4}$ and $S_4$}\\

10 & 0.3621 & 0.0846 & \textbf{0.7101} & 0.1181\\
\rowcolor{rowgray}
15 & 0.4502 & 0.1319 & \textbf{0.6737} & 0.2166\\
20 & 0.5025 & 0.2378 & \textbf{0.6656} & 0.3353\\


\rowcolor{headerblue}
\multicolumn{5}{c}{$S_{<5}$ and $S_5$}\\

10 & \textbf{0.6284} & 0.1178 & 0.3822 & 0.1515\\
\rowcolor{rowgray}
15 & \textbf{0.7101} & 0.1486 & 0.4619 & 0.2233\\
20 & \textbf{0.6985} & 0.2426 & 0.5653 & 0.2418\\


\rowcolor{headerblue}
\multicolumn{5}{c}{$S_{<6}$ and $S_6$}\\

10 & \textbf{0.7851} & 0.0907 & 0.4895 & 0.0713\\
\rowcolor{rowgray}
15 & \textbf{0.7736} & 0.1507 & 0.5002 & 0.1382\\
20 & \textbf{0.8389} & 0.2268 & 0.5108 & 0.1692\\

\bottomrule
\end{tabular}
\end{table}

\section{Appendix to Experiments}
\label{app:experiments}

\paragraph{Datasets.}
We conduct our experiments on GSM8K~\citep{cobbe2021gsm8k},  PRMBench~\citep{song2025prmbench}, and AIME2025~\citep{aime25}.

\textbf{GSM8K}~\citep{cobbe2021gsm8k} is a benchmark consisting of grade-school 
mathematical word problems that require multi-step reasoning. It contains approximately $7.47$K training and $1.32k$ test mathematics word problems. Solving each problem typically requires multiple steps of numerical reasoning, including intermediate arithmetic operations, making GSM8K a widely used benchmark for evaluating the chain-of-thought reasoning capabilities of language models. We use GSM8K to analyze step-wise reasoning patterns, study the relationship between RRI and reasoning length, and conduct our RRI-guided fine-tuning experiments.

\textbf{PRMBench}~\citep{song2025prmbench} is a process-level reasoning benchmark designed to evaluate the ability to identify errors occurring at individual steps of a reasoning trajectory. The benchmark contains fine-grained annotations covering different categories of reasoning errors. In our experiments, we specifically focus on the \emph{redundancy} category to evaluate the ability of RRI to identify repetitive reasoning steps. This class has around $750$ problems. We treat steps annotated as redundant as the positive class and non-redundant steps as the negative class when computing the evaluation metric. 

\textbf{AIME2025}~\citep{aime25} is a benchmark containing 30 problems from the 2025 American Invitational Mathematics Examination I and II. Compared with GSM8K, its problems are more complex and require greater algebraic manipulation, geometric insight, and symbolic reasoning.

\paragraph{Spearman Correlation.} We use Spearman's rank correlation coefficient~\citep{zar2005spearman} to quantify the relationship between reasoning repetitiveness and reasoning length. For two variables $X$ and $Y$, Spearman's correlation is defined as the Pearson correlation between their ranks:
\begin{equation}
    \rho =
    \frac{
        \operatorname{Cov}\!\left(\operatorname{rank}(X),
        \operatorname{rank}(Y)\right)
    }{
        \sigma_{\operatorname{rank}(X)}
        \sigma_{\operatorname{rank}(Y)}
    },
\end{equation}
where $\operatorname{rank}(X)$ and $\operatorname{rank}(Y)$ denote the 
rank-transformed variables, $\operatorname{Cov}(\cdot,\cdot)$ denotes 
covariance, and $\sigma$ denotes standard deviation. The value ranges from $-1$ (perfect negative monotonicity) to $1$ (perfect positive monotonicity), with $0$ indicating no monotonic relationship.

\subsection{PID measures under synthetic noise} 
\label{app:noise}\begin{figure}[!htbp]
        \centering
        \includegraphics[width=0.9\linewidth]{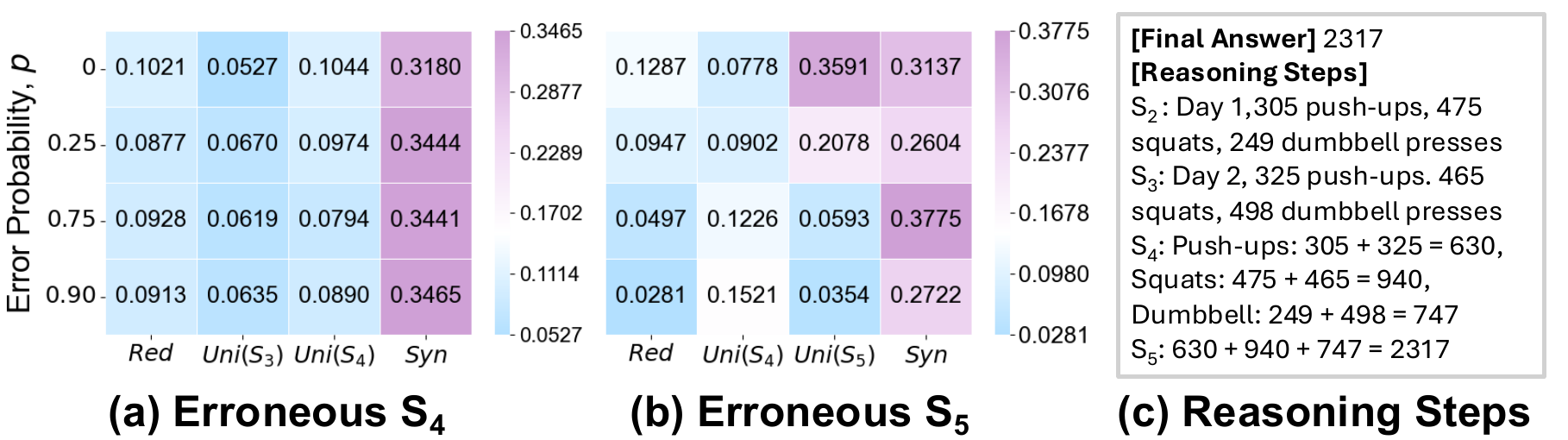}
    \caption{Interpretable trends of PID measures under errors for word problem: (a,b) Under errors, uniqueness Uni($S_i$) corresponding to a correct step increases with $p$, reflecting that correct steps contribute increasingly distinct information relative to their erroneous counterparts. Conversely, Uni($S_i$) associated with the erroneous step decreases as error probability $p$ increases. (c) Reference example for the reasoning steps in concise form.
    \label{fig:cot2} } 
    \end{figure}

We consider a word problem from the GSM8K dataset: \emph{``Darren decides to do body exercises for a whole week. He does $\{x\}$ pushups, $\{y\}$ squats, and $\{z\}$ dumbbell presses on the first day. On the second day, he does 20 more pushups than on the first day, ten fewer squats, and doubles the number of dumbbell presses. What's the total count of the activities he's done in the two days?''} We prompt DeepSeek-R1-Distill-Qwen-32B to generate a model response for the problem and retain the text preceding the \texttt{</think>} tag as the reasoning trajectory for our analysis. We manually segment each reasoning trajectory at double-newline (\texttt{\textbackslash n\textbackslash n}) boundaries. We then generate 1,000 problem instances by uniformly sampling $x,y \in [10,500]$ and $z \in [10,100]$. To study the effect of reasoning errors, we perturb the total push-up calculation in step $S_4$ with probability $p$ by adding uniformly sampled integer noise from $[0,1000]$. We then analyze the resulting PID measures using \strategy, as shown in Fig.~\ref{fig:cot2}. Consistent with Theorem~\ref{thm:uniqueness_noise}, under errors, the uniqueness $\mathrm{Uni}(S_i)$ corresponding to a correct step increases with $p$, indicating that correct steps provide increasingly distinct information relative to their erroneous counterparts. In contrast, the uniqueness associated with the erroneous step decreases as the error probability $p$ increases. Fig.~\ref{fig:cot2}(b) shows that introduced errors propagate toward the final step $S_5$. For that step, the drop in $\mathrm{Uni}(S_5)$ is significant since this step directly includes the final answer $A$. Thus, \strategy can help identify incorrectness in reasoning.

\subsection{Step-wise Repetitiveness} 

\emph{Setup.} To assess whether Step-RRI captures true repetitiveness in a reasoning trajectory, we consider three word problems from the GSM8K dataset. For each problem, we first obtain reasoning trajectories from DeepSeek-R1-Distill-Qwen-32B and QwQ-32B. We then generate 1,000 structure-preserving samples by parameterizing selected numerical values in the problem and updating the corresponding quantities in the reasoning trajectory while preserving its reasoning structure. For Problem~1, `Alisa biked \{x\} miles per hour for 4.5 hours. Stanley biked at \{y\} miles per hour for 2.5 hours. How many miles did Alisa and Stanley bike in total?'', we independently sample $x$ and $y$ uniformly from the integer range $[10,100]$ and update the corresponding calculations in the trajectory. For Problem~2, `'Anthony had {x} pencils. He gave 1/2 of his pencils to Brandon, and 3/5 of the remaining pencils to Charlie. He kept the remaining pencils. How many pencils did Anthony keep?'', and for Problem~3, ``James runs {x} miles a day for 5 days a week. If he runs 10 miles an hour, how many hours does he run a week?'', we similarly sample $x$ uniformly from the integer range $[10,1000]$ and update all dependent quantities in the reasoning trajectory. Then, we apply \strategy to get the RRI measures.

\textbf{Step-RRI for DeepSeek-R1-Distill-Qwen-32B.} Fig.~\ref{fig:r1_pid} shows the step-wise PID decomposition together with the resulting Step-RRI values (with $\eta = 1$) for reasoning trajectories generated by DeepSeek-R1-Distill-Qwen-32B. For Problems~1 and~2, Step-RRI remains negative or close to zero across the reasoning steps, indicating that the current steps primarily contribute new or complementary answer-relevant information relative to the preceding reasoning. In Problem~3, Step-RRI becomes positive at the final step, indicating that the current step increasingly recovers answer-relevant information already established by the preceding reasoning. These examples illustrate how Step-RRI summarizes the relative balance of redundant, unique, and synergistic information to localize true repetitive reasoning within a trajectory.
\label{app:contribution2}
\begin{figure}[!htbp]
    \centering
        \includegraphics[width=\linewidth]{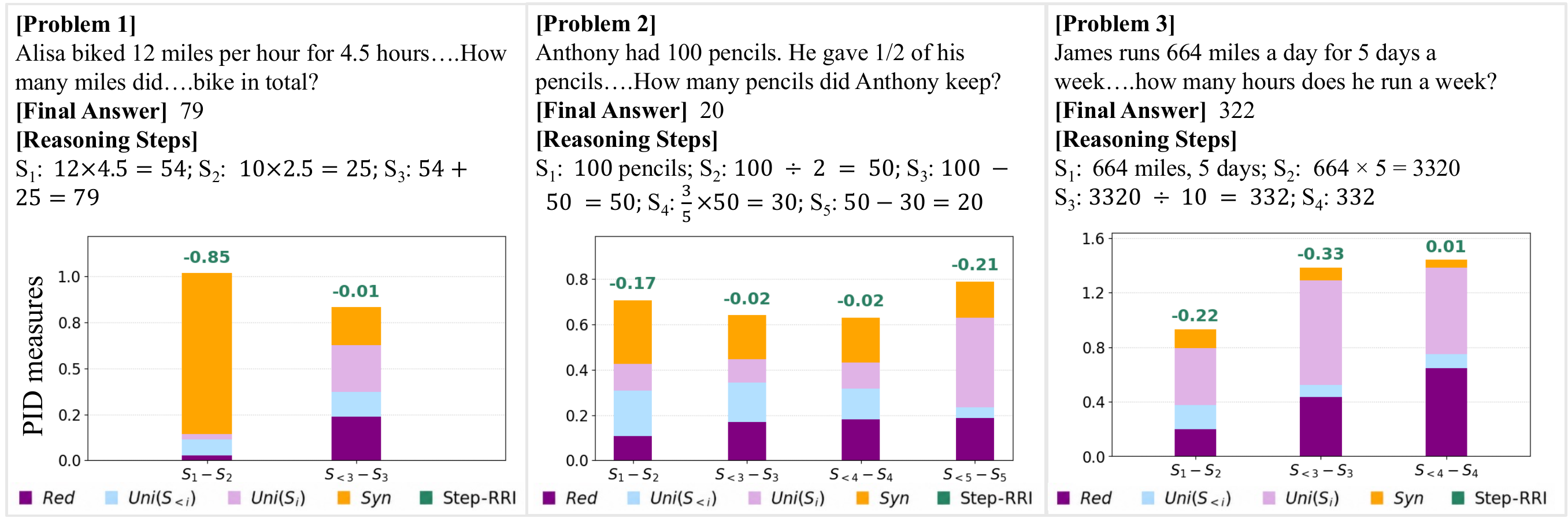}
\caption{\textbf{Step-RRI for DeepSeek-R1-Distill-Qwen-32B reasoning trajectories.}
Step-wise PID components and the corresponding Step-RRI values for three representative GSM8K reasoning trajectories. Negative Step-RRI indicates that the current step contributes primarily new or complementary answer-relevant information, whereas positive Step-RRI indicates increased repetitiveness with respect to the preceding reasoning.}
\label{fig:r1_pid}

\end{figure}
 
Next, we provide one representative full reasoning trajectory for each model, QwQ-32B and DeepSeek-R1-Distill-Qwen-32B. These examples show the original natural-language reasoning traces used in our experiments.
\begin{myquote}
\textbf{[Question]}
\vspace{1em}

James runs 664 miles a day for 5 days a week. If he runs 10 miles an hour how many hours does he run a week?
\vspace{1em}

\textbf{[Reasoning of QwQ-32B]}
\begin{enumerate}[label=\arabic*.,leftmargin=15pt, itemsep=2pt, topsep=5pt, parsep=0pt, partopsep=0pt]
\item Okay, let's see here. James runs 664 miles a day for 5 days a week. And we need to find out how many hours he runs each week, given that his speed is 10 miles per hour.
\item First, I think I should figure out the total distance he runs in a week. Since he does 664 miles each day and runs 5 days a week, I can multiply 664 by 5. Let me do that: 664 times 5 equals... 3320. So he runs 3320 miles in a week.
\item Now, the question is asking for the time he spends running, not the distance. His speed is 10 miles per hour. Hmm, so if he runs at 10 mph, then the time taken would be distance divided by speed. Right? Because time equals distance over speed.
\item So the total time per week would be the total miles divided by his speed. That would be 3320 miles divided by 10 miles per hour. Let me calculate that: 3320 divided by 10 is 332.0. So that gives 332.0 hours.
\item Wait, let me double-check. If he runs 664 miles a day at 10 mph, how long does each day's run take? For one day: 664 miles divided by 10 mph would be 66.4 hours per day. Then, multiplying that by 5 days: 66.4 times 5 equals 332.0 hours. Yep, same answer. So that seems right.
\item I don't think there's any trick here. Maybe I'm missing something? Like, does he take rest days or anything? The problem says he runs 5 days a week, so the other two days aren't counted. The calculation for total miles is correct. Then dividing by speed gives total hours. Both methods give 332.0, so I think that's solid. Alright, the answer should be 332.0 hours per week.
\end{enumerate}

\textbf{[Reasoning of DeepSeek-R1-Distill-Qwen-32B]}

\begin{enumerate}[label=\arabic*.,leftmargin=15pt, itemsep=2pt, topsep=5pt, parsep=0pt, partopsep=0pt]
\item First, I need to determine how many miles James runs each week. He runs 664 miles per day and does this for 5 days.
\item So, multiplying the daily distance by the number of days gives me the weekly total:
664 miles/day * 5 days = 3320 miles/week.
\item Next, since James maintains a running speed of 10 miles per hour, I can calculate the total time he spends running each week by dividing the total weekly distance by his speed: 3320 miles / 10 miles/hour = 332.0 hours.
\item Therefore, James runs a total of 332.0 hours each week.
\end{enumerate}
\end{myquote}
\textbf{Baselines.} We adapt two complementary measures as baselines: Cosine Similarity, which captures semantic overlap between reasoning steps, and InfoGain, which measures the reduction in predictive uncertainty contributed by a reasoning step. Intuitively, if the current reasoning step repeats or conveys the same information as the preceding reasoning, we expect it to exhibit high cosine similarity due to greater semantic overlap, while providing little additional information and therefore exhibiting lower information gain.

\emph{Cosine Similarity.} As a similarity-based baseline, we measure the semantic overlap between the current reasoning step $S_i$ and the preceding reasoning $S_{<i}$. We encode $S_i$ and $S_{<i}$ separately using  all-MiniLM-L6-v2~\citep{reimers2019sentence} and normalize the resulting
embeddings. Their cosine similarity is then given by
\begin{equation}
    \operatorname{CosSim}(S_i,S_{<i})
    =
    \mathbf{e}_{S_i}^{\top}\mathbf{e}_{S_{<i}},
\end{equation}
where $\mathbf{e}_{S_i}$ and $\mathbf{e}_{S_{<i}}$ denote the normalized
embeddings of the current and preceding reasoning, respectively. A high
cosine similarity indicates substantial semantic overlap with previously
established reasoning and therefore provides a natural proxy for
repetitiveness.

\emph{InfoGain.}
As an information-gain baseline, we adapt the information-gain principle
of~\citet{yong2025think}, where an informative reasoning step is expected
to reduce predictive uncertainty. At each step $S_i$, we use QwQ-32B to
obtain the top-$K$ next-token probabilities conditioned on the question $Q$
and the reasoning prefix $S_{1:i}$. Since only the top-$K$ probabilities are
retained, we renormalize them as
\begin{equation}
    \tilde{P}(v_k \mid Q,S_{1:i})
    =
    \frac{P(v_k \mid Q,S_{1:i})}
    {\sum_{j=1}^{K} P(v_j \mid Q,S_{1:i})},
\end{equation}
where $v_k$ denotes the $k$-th most probable next token. We then compute the
entropy of the normalized top-$K$ distribution as
\begin{equation}
    H_i =
    -\sum_{k=1}^{K}
    \tilde{P}(v_k \mid Q,S_{1:i})
    \log_2 \tilde{P}(v_k \mid Q,S_{1:i}),
\end{equation}
and define the information gain associated with $S_i$ as
\begin{equation}
    \Delta I_i = H_{i-1}-H_i.
\end{equation}
We use $K=10$ in our experiments. A large positive $\Delta I_i$ indicates that incorporating $S_i$ substantially reduces next-token uncertainty, whereas a small or negative $\Delta I_i$ suggests that the step provides limited additional predictive information. We therefore use low information gain as a proxy for repetitive reasoning, as a repetitive step is expected to contribute limited new predictive information.


\subsection{Trajectory-RRI versus Reasoning Length}
\label{app:contribution3}

\emph{Setup.} For our analysis on GSM8K~\citep{cobbe2021gsm8k}, we first identify the problems from the training set for which both QwQ-32B and DeepSeek-R1-Distill-Qwen-32B reach the correct final numerical answer. From this common set of correctly solved problems, we sample 1,000 problem indices and collect the corresponding responses from both models, ensuring that our analysis focuses on repetitiveness among successful reasoning trajectories. Each model response consists of a thinking trajectory followed by the final response, separated by the \texttt{</think>} tag. We extract the text preceding this tag as the reasoning trajectory and use only this portion throughout the analysis. For GPT-4.1, we use the same 1,000 problems and treat its generated responses as the corresponding reasoning trajectories. For AIME2025~\citep{aime25}, we select problems for which all three models produce the correct final answer.

We apply \strategy independently to the extracted reasoning trajectories of each model to estimate Step-RRI for individual reasoning steps and subsequently compute Trajectory-RRI for each problem. To analyze the relationship between Trajectory-RRI and the true reasoning length, we group the trajectories into quantile bins according to their reasoning length. Within each bin, we report the average Trajectory-RRI and the average reasoning length. For GSM8K, we use 50 bins, and for AIME2025, we use 7 bins. Importantly, the reasoning length is computed from the original reasoning trajectory generated by the LLMs, rather than from the structure-preserving numerical variants used by \strategy for estimating the PID components. Thus, the reported reasoning length directly reflects the token counts of the original model responses.

\subsection{Trajectory-RRI Guided Data Selection}
\label{app:NetRRI}
\begin{table}[!htbp]
\centering
\caption{Test accuracy on GSM8K across different average Trajectory-RRI levels.}
\label{tab:rri_accuracy}
\begin{tabular}{cc}
\toprule
\textbf{Avg. Trajectory-RRI} & \textbf{Accuracy (\%)} \\
\midrule
3.50 & 88.55 \\
3.34 & 89.31 \\
3.19 & 89.46 \\
3.02 & 89.31 \\
2.84 & 88.78 \\
2.42 & 89.16 \\
2.04 & 88.48 \\
1.88 & 87.41 \\
1.71 & 89.16 \\
\bottomrule
\end{tabular}
\end{table}
\emph{Trajectory-RRI-Guided Dataset Construction.}
We construct fine-tuning datasets systematically with varying levels of average Trajectory-RRI using 1,000 problems from GSM8K. For each problem, we obtain two candidate reasoning trajectories: one generated by QwQ-32B and the other by DeepSeek-R1-Distill-Qwen-32B. We compute the Trajectory-RRI of each candidate trajectory using \strategy, resulting in a pair of trajectories with potentially different levels of repetitive reasoning for each problem. To control the average Trajectory-RRI of the resulting training data while keeping the set of problems fixed, we construct each fine-tuning dataset by selecting exactly one of the two candidate trajectories for every problem. Specifically, for each problem, we identify the candidate trajectory with the higher Trajectory-RRI and select it with probability $p$. This procedure yields a training set containing 1,000 problem with different average Trajectory-RRI for each value of $p$. For example, when $p=0$, the lower-Trajectory-RRI trajectory is selected for every problem, whereas $p=1$ selects the higher-Trajectory-RRI trajectory for every problem. Intermediate values of $p$ produce mixtures of the two, allowing us to systematically vary the average Trajectory-RRI of the training data without changing the underlying set of GSM8K problems. These datasets are subsequently used to fine-tune the same base model under an identical training configuration, allowing us to examine the relationship between the repetitiveness of the training trajectories defined by our measure, Trajectory-RRI, and the actual reasoning behavior of the resulting fine-tuned models.

\textbf{LLM-as-a-judge efficiency score.} We use Qwen2.5-72B-Instruct~\citep{qwen2.5} as an LLM judge to evaluate the reasoning efficiency of the fine-tuned models on the test set. Specifically, the judge assesses each generated reasoning trajectory per question for conciseness, loops, and redundant reasoning. For each fine-tuned model, we report the overall efficiency as the average efficiency score across all examples in the test set. We use the following prompt for this evaluation:

\begin{myquote}
\textbf{[System Prompt]}

You are an expert AI judge evaluating the efficiency of reasoning traces.
Focus strictly on conciseness, loops, and redundancy. Output ONLY raw JSON.

\vspace{0.5em}
\textbf{[User Prompt]}

Evaluate the efficiency and repetitiveness of this reasoning trace.
Ignore whether the final answer is right or wrong; focus entirely on the
pacing and overhead of the logic.

\vspace{0.5em}
Question: \texttt{\{Original Question\}}
\vspace{0.5em}

Trace: \texttt{\{Reasoning Trace\}}
\vspace{0.5em}

\textbf{Score Guidelines (0--10 Efficiency Score):}

\begin{itemize}
    \item \textbf{10:} Maximally concise. Direct, optimal path with zero
    bloat or repetition.

    \item \textbf{7--9:} Highly efficient, but contains minor phrasing
    redundancy or slight over-explanation.

    \item \textbf{4--6:} Noticeably inefficient. Contains repetitive
    steps, circular logic, or heavy fluff.

    \item \textbf{1--3:} Extremely bloated, highly repetitive, or stuck
    in reasoning loops.

    \item \textbf{0:} Total failure of efficiency (e.g., endless
    repeating strings or infinite loops).
\end{itemize}

\textbf{Output Format (JSON only, no markdown):}

\begin{verbatim}
{
  "efficiency_score": 0-10,
  "reason": "One-sentence analysis of redundancy and pacing."
}
\end{verbatim}


\end{myquote}

\textbf{Fine-tuning setup.} We fine-tune Qwen2.5-7B-Instruct using Low-Rank Adaptation (LoRA)~\citep{hu2021lora}. We use a LoRA rank of $r=32$,
a scaling factor of $\alpha=32$, and a LoRA dropout of $0.05$.
LoRA adapters are applied to all attention
and MLP projection layers. We train each model for $10$ epochs using the AdamW optimizer with a learning rate of $1\times10^{-5}$, weight decay of
$1\times10^{-4}$, and momentum parameters $\beta_1=0.9$ and $\beta_2=0.95$. We use a cosine learning-rate scheduler with a warmup ratio of $0.05$. The per-device batch size is set to $1$, with gradient accumulation of $8$ steps, resulting
in an effective global batch size of $32$, where the number of GPUs used for training is $4$. We set the maximum sequence length to $8192$ tokens and use bfloat16 precision.

For distributed training, we use Fully Sharded Data Parallel
(FSDP) with full parameter sharding and automatic module wrapping.
We additionally enable gradient checkpointing to reduce memory
consumption. Model checkpoints are saved after each training epoch,
and no intermediate evaluation is performed during training.
All models are trained using the same optimization and fine-tuning
configuration to ensure a controlled comparison across datasets
with average Trajectory-RRI values.



\subsection{Runtime and Resource Usage}
\label{app:runtime}

All experiments are conducted using NVIDIA RTX 6000 GPUs with CUDA 12.8. For GSM8K, \strategy takes an average of approximately 181 seconds per sample for QwQ-32B and 82 seconds per sample for DeepSeek-R1-Distill-Qwen-32B. For PRMBench, the average runtime of \strategy is approximately 38 seconds per sample. For AIME2025, the average runtimes are approximately 609 and 320 seconds per sample for QwQ-32B and DeepSeek-R1-Distill-Qwen-32B, respectively.

\end{document}

%% file: math_commands.tex
\usepackage{amsmath,amsfonts,bm}

\def\eqref#1{equation~\ref{#1}}

\def\1{\bm{1}}

\DeclareMathAlphabet{\mathsfit}{\encodingdefault}{\sfdefault}{m}{sl}
\SetMathAlphabet{\mathsfit}{bold}{\encodingdefault}{\sfdefault}{bx}{n}



%% file: preamble.tex
\usepackage{xspace}
\newcommand{\strategy}{\textsc{Slider}\xspace}
\usepackage{titletoc}

%% file: main.bbl
\begin{thebibliography}{43}
\providecommand{\natexlab}[1]{#1}
\providecommand{\url}[1]{\texttt{#1}}
\expandafter\ifx\csname urlstyle\endcsname\relax
  \providecommand{\doi}[1]{doi: #1}\else
  \providecommand{\doi}{doi: \begingroup \urlstyle{rm}\Url}\fi

\bibitem[Ali et~al.(2025)Ali, Caso, Irwin, and Li{\`o}]{ali2025entropy}
Riccardo Ali, Francesco Caso, Christopher Irwin, and Pietro Li{\`o}.
\newblock Entropy-lens: The information signature of transformer computations.
\newblock \emph{arXiv preprint arXiv:2502.16570}, 2025.

\bibitem[Banerjee et~al.(2018)Banerjee, Olbrich, Jost, and Rauh]{banerjee2018unique}
Pradeep~Kr Banerjee, Eckehard Olbrich, J{\"u}rgen Jost, and Johannes Rauh.
\newblock Unique informations and deficiencies.
\newblock In \emph{2018 56th Annual Allerton Conference on Communication, Control, and Computing (Allerton)}, pp.\  32--38, 2018.

\bibitem[Bertschinger et~al.(2014)Bertschinger, Rauh, Olbrich, Jost, and Ay]{bertschinger2014quantifying}
Nils Bertschinger, Johannes Rauh, Eckehard Olbrich, J{\"u}rgen Jost, and Nihat Ay.
\newblock Quantifying unique information.
\newblock \emph{Entropy}, 16\penalty0 (4):\penalty0 2161--2183, 2014.

\bibitem[Cobbe et~al.(2021)Cobbe, Kosaraju, Bavarian, Chen, Jun, Kaiser, Plappert, Tworek, Hilton, Nakano, Hesse, and Schulman]{cobbe2021gsm8k}
Karl Cobbe, Vineet Kosaraju, Mohammad Bavarian, Mark Chen, Heewoo Jun, Lukasz Kaiser, Matthias Plappert, Jerry Tworek, Jacob Hilton, Reiichiro Nakano, Christopher Hesse, and John Schulman.
\newblock Training verifiers to solve math word problems.
\newblock \emph{arXiv preprint arXiv:2110.14168}, 2021.

\bibitem[Cover \& Thomas(2012)Cover and Thomas]{cover2012elements}
Thomas~M Cover and Joy~A Thomas.
\newblock \emph{{Elements of Information Theory}}.
\newblock John Wiley \& Sons, 2012.

\bibitem[Dewan et~al.(2024)Dewan, Zawar, Saxena, Chang, Luo, and Bisk]{diffusionPID}
Shaurya Dewan, Rushikesh Zawar, Prakanshul Saxena, Yingshan Chang, Andrew Luo, and Yonatan Bisk.
\newblock Diffusion pid: Interpreting diffusion via partial information decomposition.
\newblock \emph{Advances in Neural Information Processing Systems}, 37:\penalty0 2045--2079, 2024.

\bibitem[Dissanayake et~al.(2024)Dissanayake, Hamman, Halder, Sucholutsky, Zhang, and Dutta]{dissanayake2024quantifying}
Pasan Dissanayake, Faisal Hamman, Barproda Halder, Ilia Sucholutsky, Qiuyi Zhang, and Sanghamitra Dutta.
\newblock Quantifying knowledge distillation using partial information decomposition.
\newblock In \emph{International Conference on Artificial Intelligence and Statistics}, 2024.

\bibitem[Dutta \& Hamman(2023)Dutta and Hamman]{entropy_survey}
Sanghamitra Dutta and Faisal Hamman.
\newblock A review of partial information decomposition in algorithmic fairness and explainability.
\newblock \emph{Entropy}, 25\penalty0 (5), 2023.
\newblock ISSN 1099-4300.
\newblock \doi{10.3390/e25050795}.
\newblock URL \url{https://www.mdpi.com/1099-4300/25/5/795}.

\bibitem[Dutta et~al.(2020)Dutta, Venkatesh, Mardziel, Datta, and Grover]{dutta2020information}
Sanghamitra Dutta, Praveen Venkatesh, Piotr Mardziel, Anupam Datta, and Pulkit Grover.
\newblock An information-theoretic quantification of discrimination with exempt features.
\newblock In \emph{Proceedings of the AAAI Conference on Artificial Intelligence}, volume~34, pp.\  3825--3833, 2020.

\bibitem[Dutta et~al.(2021)Dutta, Venkatesh, Mardziel, Datta, and Grover]{dutta2021fairness}
Sanghamitra Dutta, Praveen Venkatesh, Piotr Mardziel, Anupam Datta, and Pulkit Grover.
\newblock Fairness under feature exemptions: Counterfactual and observational measures.
\newblock \emph{IEEE Transactions on Information Theory}, 67\penalty0 (10):\penalty0 6675--6710, 2021.
\newblock \doi{10.1109/TIT.2021.3103206}.

\bibitem[Ehrlich et~al.(2022)Ehrlich, Schneider, Wibral, Priesemann, and Makkeh]{ehrlich2022partial}
David~A Ehrlich, Andreas~C Schneider, Michael Wibral, Viola Priesemann, and Abdullah Makkeh.
\newblock Partial information decomposition reveals the structure of neural representations.
\newblock \emph{arXiv preprint arXiv:2209.10438}, 2022.

\bibitem[Farquhar et~al.(2024)Farquhar, Kossen, Kuhn, and Gal]{farquhar2024detecting}
Sebastian Farquhar, Jannik Kossen, Lorenz Kuhn, and Yarin Gal.
\newblock Detecting hallucinations in large language models using semantic entropy.
\newblock \emph{Nature}, 630\penalty0 (8017):\penalty0 625--630, 2024.

\bibitem[Gan et~al.(2025)Gan, Liao, and Liu]{gan2025rethinking}
Zeyu Gan, Yun Liao, and Yong Liu.
\newblock Rethinking external slow-thinking: From snowball errors to probability of correct reasoning.
\newblock \emph{arXiv preprint arXiv:2501.15602}, 2025.

\bibitem[Goswami et~al.(2023)Goswami, Merkley, and Grover]{goswami2023computing}
Chaitanya Goswami, Amanda Merkley, and Pulkit Grover.
\newblock Computing unique information for poisson and multinomial systems.
\newblock \emph{arXiv preprint arXiv:2305.07013}, 2023.

\bibitem[Guo et~al.(2025)Guo, Yang, Zhang, Song, Wang, Zhu, Xu, Zhang, Ma, Bi, et~al.]{guo2025deepseek}
Daya Guo, Dejian Yang, Haowei Zhang, Junxiao Song, Peiyi Wang, Qihao Zhu, Runxin Xu, Ruoyu Zhang, Shirong Ma, Xiao Bi, et~al.
\newblock Deepseek-r1: Incentivizing reasoning capability in llms via reinforcement learning.
\newblock \emph{arXiv preprint arXiv:2501.12948}, 2025.

\bibitem[Halder et~al.(2025{\natexlab{a}})Halder, Dissanayake, and Dutta]{halder2025learning}
Barproda Halder, Pasan Dissanayake, and Sanghamitra Dutta.
\newblock Learning invariant graph representations through redundant information.
\newblock \emph{arXiv preprint arXiv:2512.06154}, 2025{\natexlab{a}}.

\bibitem[Halder et~al.(2025{\natexlab{b}})Halder, Hamman, Dissanayake, Zhang, Sucholutsky, and Dutta]{haldertowards}
Barproda Halder, Faisal Hamman, Pasan Dissanayake, Qiuyi Zhang, Ilia Sucholutsky, and Sanghamitra Dutta.
\newblock Towards formalizing spuriousness of biased datasets using partial information decomposition.
\newblock \emph{Transactions on Machine Learning Research}, 2025{\natexlab{b}}.

\bibitem[Hamman \& Dutta(2024{\natexlab{a}})Hamman and Dutta]{hamman2024demystifying}
Faisal Hamman and Sanghamitra Dutta.
\newblock Demystifying local and global fairness trade-offs in federated learning using partial information decomposition.
\newblock \emph{International Conference on Learning Representations}, 2024{\natexlab{a}}.

\bibitem[Hamman \& Dutta(2024{\natexlab{b}})Hamman and Dutta]{hammanUnified}
Faisal Hamman and Sanghamitra Dutta.
\newblock A unified view of group fairness tradeoffs using partial information decomposition.
\newblock In \emph{2024 IEEE International Symposium on Information Theory (ISIT)}, pp.\  214--219, 2024{\natexlab{b}}.
\newblock \doi{10.1109/ISIT57864.2024.10619698}.

\bibitem[Hu et~al.(2021)Hu, Shen, Wallis, Allen-Zhu, Li, Wang, Wang, and Chen]{hu2021lora}
Edward~J Hu, Yelong Shen, Phillip Wallis, Zeyuan Allen-Zhu, Yuanzhi Li, Shean Wang, Lu~Wang, and Weizhu Chen.
\newblock Lora: Low-rank adaptation of large language models.
\newblock \emph{arXiv preprint arXiv:2106.09685}, 2021.

\bibitem[Kossen et~al.(2024)Kossen, Han, Razzak, Schut, Malik, and Gal]{kossen2024semantic}
Jannik Kossen, Jiatong Han, Muhammed Razzak, Lisa Schut, Shreshth Malik, and Yarin Gal.
\newblock Semantic entropy probes: Robust and cheap hallucination detection in llms.
\newblock \emph{arXiv preprint arXiv:2406.15927}, 2024.

\bibitem[Liang et~al.(2023)Liang, Ling, Cheng, Obolenskiy, Liu, Pandey, Wilf, Morency, and Salakhutdinov]{liang2023multimodal}
Paul~Pu Liang, Chun~Kai Ling, Yun Cheng, Alex Obolenskiy, Yudong Liu, Rohan Pandey, Alex Wilf, Louis-Philippe Morency, and Ruslan Salakhutdinov.
\newblock Multimodal learning without labeled multimodal data: Guarantees and applications.
\newblock \emph{arXiv preprint arXiv:2306.04539}, 2023.

\bibitem[Lyu et~al.(2024)Lyu, Clark, and Raviv]{lyu2024explicit}
Aobo Lyu, Andrew Clark, and Netanel Raviv.
\newblock Explicit formula for partial information decomposition.
\newblock In \emph{2024 IEEE International Symposium on Information Theory (ISIT)}, pp.\  2329--2334, 2024.
\newblock \doi{10.1109/ISIT57864.2024.10619369}.

\bibitem[Mohamadi et~al.(2023)Mohamadi, Doretto, and Adjeroh]{mohamadi2023more}
Salman Mohamadi, Gianfranco Doretto, and Donald~A Adjeroh.
\newblock More synergy, less redundancy: Exploiting joint mutual information for self-supervised learning.
\newblock \emph{arXiv preprint arXiv:2307.00651}, 2023.

\bibitem[OpenAI.(2024)]{OpenAI}
OpenAI.
\newblock Learning to reason with {LLMs}, 2024.
\newblock URL \url{https://openai.com/index/learning-to-reason-with-llms/}.

\bibitem[OpenAI(2024)]{openai2024gpt4omini}
OpenAI.
\newblock Gpt-4o mini: Advancing cost-efficient intelligence.
\newblock \url{https://openai.com/index/gpt-4o-mini-advancing-cost-efficient-intelligence/}, July 2024.
\newblock Accessed: 2026-08-06.

\bibitem[{OpenAI}(2025{\natexlab{a}})]{openai2025gpt41}
{OpenAI}.
\newblock Introducing gpt-4.1 in the api.
\newblock \url{https://openai.com/index/gpt-4-1/}, April 2025{\natexlab{a}}.
\newblock Accessed: 2026-09-25.

\bibitem[{OpenAI}(2025{\natexlab{b}})]{openai2025gpt5mini}
{OpenAI}.
\newblock {GPT-5 mini}.
\newblock \url{https://developers.openai.com/api/docs/models/gpt-5-mini}, 2025{\natexlab{b}}.
\newblock OpenAI API model documentation.

\bibitem[Pakman et~al.(2021)Pakman, Nejatbakhsh, Gilboa, Makkeh, Mazzucato, Wibral, and Schneidman]{pakman2021estimating}
Ari Pakman, Amin Nejatbakhsh, Dar Gilboa, Abdullah Makkeh, Luca Mazzucato, Michael Wibral, and Elad Schneidman.
\newblock Estimating the unique information of continuous variables.
\newblock \emph{Advances in neural information processing systems}, 34:\penalty0 20295--20307, 2021.

\bibitem[Reimers \& Gurevych(2019)Reimers and Gurevych]{reimers2019sentence}
Nils Reimers and Iryna Gurevych.
\newblock Sentence-bert: Sentence embeddings using siamese bert-networks.
\newblock In \emph{Proceedings of the 2019 conference on empirical methods in natural language processing and the 9th international joint conference on natural language processing (EMNLP-IJCNLP)}, pp.\  3982--3992, 2019.

\bibitem[Song et~al.(2025)Song, Su, Qu, Zhou, and Cheng]{song2025prmbench}
Mingyang Song, Zhaochen Su, Xiaoye Qu, Jiawei Zhou, and Yu~Cheng.
\newblock Prmbench: A fine-grained and challenging benchmark for process-level reward models.
\newblock In \emph{Proceedings of the 63rd Annual Meeting of the Association for Computational Linguistics (Volume 1: Long Papers)}, pp.\  25299--25346, 2025.

\bibitem[Tax et~al.(2017)Tax, Mediano, and Shanahan]{tax2017partial}
Tycho Tax, Pedro Mediano, and Murray Shanahan.
\newblock The partial information decomposition of generative neural network models.
\newblock \emph{Entropy}, 19\penalty0 (9):\penalty0 474, 2017.

\bibitem[Team(2024)]{qwen2.5}
Qwen Team.
\newblock Qwen2.5: A party of foundation models, September 2024.
\newblock URL \url{https://qwenlm.github.io/blog/qwen2.5/}.

\bibitem[Team(2025)]{team2025qwq}
Qwen Team.
\newblock Qwq-32b: Embracing the power of reinforcement learning, 2025.

\bibitem[Ton et~al.(2024)Ton, Taufiq, and Liu]{ton2024understanding}
Jean-Francois Ton, Muhammad~Faaiz Taufiq, and Yang Liu.
\newblock Understanding chain-of-thought in llms through information theory.
\newblock \emph{arXiv preprint arXiv:2411.11984}, 2024.

\bibitem[Venkatesh et~al.(2024)Venkatesh, Bennett, Gale, Ramirez, Heller, Durand, Olsen, and Mihalas]{venkatesh2024gaussian}
Praveen Venkatesh, Corbett Bennett, Sam Gale, Tamina Ramirez, Greggory Heller, Severine Durand, Shawn Olsen, and Stefan Mihalas.
\newblock Gaussian partial information decomposition: Bias correction and application to high-dimensional data.
\newblock \emph{Advances in Neural Information Processing Systems}, 36, 2024.

\bibitem[Wei et~al.(2022)Wei, Wang, Schuurmans, Bosma, Xia, Chi, Le, and Zhou]{wei2022chain}
Jason Wei, Xuezhi Wang, Dale Schuurmans, Maarten Bosma, Fei Xia, Ed~Chi, Quoc~V Le, and Denny Zhou.
\newblock Chain-of-thought prompting elicits reasoning in large language models.
\newblock In \emph{Advances in Neural Information Processing Systems}, volume~35, pp.\  24824--24837, 2022.

\bibitem[Williams \& Beer(2010)Williams and Beer]{williams2010nonnegativePID}
Paul~L Williams and Randall~D Beer.
\newblock Nonnegative decomposition of multivariate information.
\newblock \emph{arXiv preprint arXiv:1004.2515}, 2010.

\bibitem[Wollstadt et~al.(2023)Wollstadt, Schmitt, and Wibral]{wollstadt2023rigorous}
Patricia Wollstadt, Sebastian Schmitt, and Michael Wibral.
\newblock A rigorous information-theoretic definition of redundancy and relevancy in feature selection based on (partial) information decomposition.
\newblock \emph{J. Mach. Learn. Res.}, 24:\penalty0 131--1, 2023.

\bibitem[Xiu et~al.(2026)Xiu, Luo, and Nakayama]{xiu2026comprehensive}
Lixin Xiu, Xufang Luo, and Hideki Nakayama.
\newblock A comprehensive information-decomposition analysis of large vision-language models.
\newblock \emph{arXiv preprint arXiv:2603.29676}, 2026.

\bibitem[Yong et~al.(2025)Yong, Zhou, Zhang, Li, Zheng, and Wu]{yong2025think}
Xixian Yong, Xiao Zhou, Yingying Zhang, Jinlin Li, Yefeng Zheng, and Xian Wu.
\newblock Think or not? exploring thinking efficiency in large reasoning models via an information-theoretic lens.
\newblock \emph{arXiv preprint arXiv:2505.18237}, 2025.

\bibitem[Zar(2005)]{zar2005spearman}
Jerrold~H Zar.
\newblock Spearman rank correlation.
\newblock \emph{Encyclopedia of biostatistics}, 7, 2005.

\bibitem[Zhang \& Team(2025)Zhang and Team]{aime25}
Yifan Zhang and Math-AI Team.
\newblock American invitational mathematics examination (aime) 2025, 2025.
\newblock URL \url{https://huggingface.co/datasets/math-ai/aime25}.

\end{thebibliography}
